\documentclass[amsart,11pt,oneside,english]{article}
\usepackage{amssymb,amsmath,amsfonts, amsmath,mathtools,
amsthm,euscript,mathrsfs,MnSymbol,verbatim,enumerate,multirow,bbding,slashed,multicol,color,array, esint,babel, tikz,tikz-cd,tikz-3dplot,tkz-graph,pgfplots,ytableau,graphicx,float,rotating,hyperref,geometry,mathdots,savesym,cite, wasysym,amscd,graphicx,pifont,float,setspace,multirow,wrapfig,picture,subfigure, amsthm, enumitem}
\usepackage[normalem]{ulem} 
\usepackage[utf8]{inputenc}
\usepackage{prettyref}
\usepackage[T1]{fontenc}
\usepackage{datetime}
\usepackage{thm-restate}

\numberwithin{equation}{section}
\usetikzlibrary{arrows,positioning,decorations.pathmorphing,   decorations.markings, matrix, patterns}
\hypersetup{colorlinks=true}
\hypersetup{linkcolor=black}
\hypersetup{citecolor=black}
\hypersetup{urlcolor=black}
\theoremstyle{plain}
\newtheorem*{thm*}{Theorem}
\theoremstyle{plain}%
\newtheorem{thm}{Theorem}[section]

\newtheorem{lem}[thm]{Lemma}
\newtheorem{cor}[thm]{Corrollary}
\newtheorem{prop}[thm]{Proposition}
\theoremstyle{definition}
\newtheorem{defn}[thm]{Definition}
	
		\newtheorem{notation}[thm]{Notation}
\newtheorem{rem}[thm]{Remark}

\newcommand{\tr}{{\bf tr} }

\numberwithin{equation}{section}

\tikzset{
  big arrow/.style={
    decoration={markings,mark=at position 1 with {\arrow[scale=1.5,#1]{>}}},
    postaction={decorate},
    shorten >=0.4pt},
  big arrow/.default=black}
\makeatother

\begin{document}

\begin{titlepage}
\begin{center}
\vspace{2cm}
{\Huge\bfseries    The Geometry of SO($3$), SO($5$), and SO($6$)-models \\  }
\vspace{2cm}
{\Large
Mboyo Esole$^{\diamondsuit}$ and Patrick Jefferson$^\clubsuit$\\}
\vspace{.6cm}
{\large $^{\diamondsuit}$ Department of Mathematics, Northeastern University}\par
{\large  360 Huntington Avenue, Boston, MA 02115, USA}\par
\vspace{.3cm}
{{\large $^\clubsuit$ Jefferson Laboratory,  Harvard University}\par
{ 17 Oxford Street, Cambridge, MA 02138, U.S.A}\par}
\vspace{2cm}
{ \bf{Abstract}}\\
\end{center}

SO($3$), SO($5$), and SO($6$)-models are 
singular elliptic fibrations with  Mordell--Weil torsion $\mathbb{Z}/2\mathbb{Z}$ and singular fibers whose dual fibers correspond to affine Dynkin diagrams of type A$_1$, C$_2$, and A$_3$ respectively, where we emphasize the distinction between SO$(n)$ and its universal cover Spin$(n)$. While the SO($3$)-model has been studied before, the SO($5$) and SO($6$)-models are studied here for the first time.
By computing crepant resolutions of their Weierstrass models, we study their fiber structures and topological invariants. 
In the special case that the SO$(n)$-model is an elliptically fibered Calabi-Yau threefold, we compute the Chern-Simons couplings and matter content of a 5D $\mathcal N=1$ supergravity theory with gauge group SO$(n)$, which is related to M-theory compactified on this Calabi-Yau threefold. 
We also verify the 6D lift of the 5D matter content is necessary and sufficient for anomaly cancellation in 6D $(1,0)$ supergravity theories geometrically engineered by F-theory compactified on the same threefold. We find that the associated 5D and 6D supergravity theories with SO$(n)$ gauge symmetry indeed differ from their Spin$(n)$ cousins, with one striking consequence of this distinction being that all such theories must include adjoint matter.

\vfill 

\end{titlepage}

\tableofcontents

\newpage

 \section{Introduction}

The utilization of gauge theory as a tool to study low-dimensional topology has dramatically increased our understanding of three-dimensional and four-dimensional manifolds. 
Similarly, gauge theories geometrically engineered by  elliptic fibrations  are nowadays the major engine of  progress in our mathematical understanding of elliptic fibrations, with key developments being  motivated by ideas and questions from theoretical physics which have helped shape our understanding of their topological invariants \cite{Euler,Char1,Char2,Arras:2016evy}, 
degenerations \cite{Bershadsky:1996nh,AE1,AE2,CDE,Clingher:2012rg,Sen:1997gv,Esole:2012tf}, the structure of their  singular fibers \cite{EY,Cattaneo,Morrison:2011mb,Braun:2013cb,EFY}, their links with the study of Higgs bundles \cite{Anderson:2017zfm}, and the  geography of their networks of flops \cite{Witten,IMS}  as discussed in \cite{G2,Esole:2014dea,SU2SU3,ESY1,ESY2,EJJN1,EJJN2,Hayashi:2014kca}. 
Conversely, M-theory and F-theory compactifications on elliptic fibrations continue to play a key role in the study of supersymmetric gauge theories with eight supercharges in six-dimensional spacetime \cite{Sadov:1996zm,Park,Morrison:2012np,Heckman:2018jxk,Haghighat:2014vxa,DelZotto:2017mee,
Monnier:2017oqd,
GM1,
Tian:2018icz,Taylor:2019ots
} and five-dimensional spacetime \cite{DelZotto:2017pti,
Jefferson:2018irk,Jefferson:2017ahm,Xie:2017pfl,Kimura:2019bzv,Grimm:2011fx,Grimm:2015zea
}.  For  reviews of F-theory, see \cite{Weigand:2018rez,Heckman:2018jxk,Park,Denef:2008wq}. 

Interestingly, the perspective that  gauge theories bring to the study of  elliptic fibrations is not only influenced by the physics of local operators associated with the  gauge algebra, but is also shaped by subtler physics considerations related to the global structure of the gauge group. An illustrative example is the distinction between the Lie groups Spin($n$) and  SO($n$), which despite being isomorphic as Lie algebras have different global structures. The global structures can typically only be distinguished in gauge theories by studying the dynamics of operators associated to non-local objects such as Wilson lines and instantons. For instance, the simple observation that the  volume of SO($3$) is half of the volume of SU($2$) has important consequences for the structure of their respective instantons as beautifully explained  in  \cite[\S 4.3]{DW}.  
The key point here is that while every SU($2$)-bundle naturally gives rise to an SO($3$)-bundle, not all SO($3$)-bundles can be lifted to SU($2$)-bundles.  This intrinsic difference is also reflected in the properties of instanton solutions, as the minimal charge of an SO($3$)-instanton can be a fraction of the minimal charge of an SU($2$)-instanton \cite{DW}.

It turns out that the distinction between Spin$(n)$ and SO$(n)$ can be seen clearly in the setting of elliptic fibrations, as the center of any gauge group geometrically engineered by an elliptic fibration is conjectured to be isomorphic to the Mordell--Weil group of the elliptic fibration \cite{Mayrhofer:2014opa}, and moreover
the $\mathbb{Z}/2\mathbb{Z}$ that appears in the exact sequence connecting Spin($n$) and SO($n$) is the torsion subgroup of the Mordell--Weil group of the elliptic fibration. Thus,  a precise understanding of the Mordell--Weil torsion subgroup is crucial for a thorough understanding of physics related to the global structure of the gauge group.

The subject of this paper is the SO$(n)$-model, and its distinction from the Spin$(n)$-model. While higher rank SO$(n)$ groups are related to I$^*_k$ fibers, the low rank examples $n=3,5,6$ are associated to Kodaira fibers which do not belong to the same infinite family and hence require a case-by-case analysis. The simple groups SO$(n)$ with $n=3,5,6$ are each subject to one of the four accidental isomorphisms of Lie algebras: 
$$
\text{A}_1\cong \text{B}_1 \cong \text{C}_1, \quad \text{B}_2\cong \text{C}_2, \quad   \text{D}_2 \cong\text{A}_1\oplus\text{A}_1, \quad \text{D}_3\cong \text{A}_3,
$$
where the semi-simple, connected and simply-connected, compact Lie groups corresponding to the above Lie algebras are the low dimensional spin groups
$$
\text{Spin}(3)\cong \text{SU}(2)\cong \text{USp}(2),\quad \text{Spin}(4)=\text{SU}(2)\times\text{SU}(2), \quad \text{Spin}(5)=\text{USp}(4), \quad \text{Spin}(6)=\text{SU}(4). 
$$
The quotients of the above spin groups by $\mathbb{Z}/2\mathbb{Z}$ are the special orthogonal groups 
$$
\text{SO}(3),\quad \text{SO}(4), \quad \text{SO}(5) \quad \text{SO}(6). 
$$
Unlike the higher rank cases, the SO$(n)$-models with $n=3,5,6$ have received comparatively less attention. While SO($3$)-models are well understood \cite{Mayrhofer:2014opa}, SO($4$)-models were only recently studied \cite{SO4}, and much less is known about SO($5$) and SO($6$)-models in F-theory; 
one goal of this paper is to fill that gap. 

We define SO($3$), SO($5$), and SO($6$)-models and study their properties through the eyes of M-theory and F-theory compactifications to five-dimensional (5D) \cite{Cadavid:1995bk,IMS}  and six-dimensional (6D) supergravity theories with eight supercharges   \cite{Ferrara:1996wv,Morrison:1996pp,Sadov:1996zm,Park}. We use the geometry of the SO$(n)$-models to explore the Coulomb branch of a 5D $\mathcal N=1$ supergravity theory with gauge group $G = \text{SO}(n)$ for $n=3,5,6$ and hypermultiplets in the representation $\textbf{R}$, where \textbf{R} is determined by the degeneration of the elliptic fibration over codimension two points in the base of the elliptic fibration. In particular, we study the one-loop exact prepotential of this 5D theory in terms of the intersection ring of the SO$(n)$-model. We also analyze the details of anomaly cancellation in the 6D $(1,0)$ theory engineered by F-theory compactified on the same SO$(n)$-model; by enumerating the charged hypermultiplets in terms of geometric quantities, we verify the same 5D content (uplifted to 6D) is both necessary and sufficient for the cancellation of anomalies in the 6D supergravity theory.

The computations described above are strong indications that the SO$(n)$-models studied in this paper ($n=3,5,6$) geometrically engineer consistent 5D and 6D supergravity theories. Unsurprisingly, it turns out that these theories are indeed distinct from their related Spin$(n)$ cousins, as their massless spectra differ in various ways. One distinction is that the geometry of the SO$(n)$-models prevents the existence of 5D and 6D hypermultiplets transforming in the spinor representation of SO$(n)$. Another distinction, somewhat more subtle, is that the structure of these SO$(n)$-models requires the existence of 5D and 6D hypermultiplets in the adjoint representation, which implies that the non-gravitational sectors of these theories cannot be described as gauge theories with ultraviolet fixed points.  

The remainder of the paper is structured as follows. In Section \ref{Sec:preliminaries}, we specify our conventions, discuss some preliminary notions, and summarize the results of our computations. Section \ref{sec:SO3} contains a detailed description of the resolution, fiber structure, and topological invariants of the SO$(3)$-model. Similar details are presented in Sections \ref{sec:SO5} and \ref{sec:SO6} for the SO$(5)$ and SO$(6)$-models, respectively. The Hodge numbers of the SO$(n)$ models are computed and tabulated in Section \ref{sec:hodge}. Finally, the geometric and topological aspects computed in the prior sections are used to study F/M-theory compactifications on special cases of SO$(n)$-models corresponding to elliptically-fibered Calabi-Yau threefolds in Section \ref{sec:stringy}.

\newpage
\begin{table}[htb]
	\begin{center}
	\scalebox{.85}{
	\begin{tabular}{|c|}
	\hline \\
	$
	\begin{array}{c}
	\arraycolsep=1.7pt\def\arraystretch{1.3}
			\begin{array}{|l|c|}
				\multicolumn{2}{c}{		\begin{matrix}	\text{\textbf{SO($3$)-model}}   \end{matrix}  } \\\hline
				\text{Weierstrass equation}  &y^2 z = x^3 + a_2 x^2 z + a_4 x z^2\\\hline
				\text{Discriminant} & \Delta=  a_4^2 (4 a_4 - a_2^2) \\\hline
				\text{Singular fibers}&\begin{array}{c} \begin{tikzpicture} 
					\draw[yshift=20,scale=.45,domain=-1.2:1.2,variable=\x, thick] plot({\x*\x-.5,\x});
					\draw[yshift=20,scale=.45,domain=-1.2:1.2,variable=\x, thick] plot({1-\x*\x-.5,\x});
					\draw[scale=.8,xshift=65,domain=-.7:.7,variable=\x, thick] plot({\x*\x,\x});
					\draw[scale=.8,xshift=65,domain=-.7:.7,variable=\x, thick] plot({-\x*\x,\x});
					\draw[yshift=60,xshift=5,scale=.6,domain=-1.2:1.2,variable=\x, thick] plot({\x*\x-1,\x*\x*\x-\x-5});
					\draw[big arrow] (.5,-.8) to (1.2,-.4);
					\draw[big arrow] (.55,1.5-.8) to (1.2,1.1-.8);
				\end{tikzpicture}\end{array}\\\hline
				\text{Matter representation} &\text{adjoint}:\  ({\bf{3}}),  \quad n_{\bf{3}}=1+6K^2\\\hline
				\text{Euler characteristic} & \frac{ 12 L}{1+ 4 L} c(TB)\\\hline
				\text{Triple intersections} &6 \mathcal F =48 L^2( -  \alpha_0^3 + \alpha_0^2 \alpha_1 +  \alpha_0 \alpha_1^2 - \alpha_1^3)\\\hline
						\multicolumn{2}{c}{		\begin{matrix}\\	\text{\textbf{SO($5$)-model}}   \end{matrix}  } \\\hline
				\text{Weierstrass equation}  & y^2 z = x^3 + a_2 x^2 z + s^2 x z^2\\\hline
				\text{Discriminant} & 	\Delta =s^4 (4 s^2 - a_2^2)\\\hline
				\text{Singular fibers} & \begin{array}{c}
					\begin{tikzpicture}[scale=1.1]
			\node[draw,circle,thick,scale=.9](C0) at (0,0) {};
			\node[draw,circle,thick,scale=.9](C1+) at (-.5,-.5) {};
			\node[draw,circle,thick,scale=.9](C1-) at (.5,-.5) {};
			\node[draw,circle,thick,scale=.9](C2) at (0,-1) {};
			\draw[xshift=5,yshift=85,scale=.45,domain=-1.2:1.2,variable=\x, thick] plot({\x*\x-1,\x*\x*\x-\x-5});
			\draw[xshift=5,yshift=15,scale=.45,domain=-1.2:1.2,variable=\x, thick] plot({\x*\x-1,\x*\x*\x-\x-5});
			\draw[thick] (C0) -- (C1+) -- (C2) -- (C1-) -- (C0);
			\draw[big arrow] (1.2,-.5) -- node[above,midway]{{}} (1.7,-.5);
			\node[draw,circle,xshift=70,thick,scale=.9](C0') at (0,0) {};
			\node[draw,circle,xshift=70,scale=.62,thick,scale=.7](C13) at (0,-.5) {\LARGE 2};
			\node[draw,circle,xshift=70,thick,scale=.9](C2) at (0,-1) {};
			\draw[dashed,scale=1.28,yshift=3,thick,red] (-.5*1.25,-.5*1.3) -- (-.5*1.25,-.5*.7) -- (.5*2.55-.5*1.3,-.5*.7) -- (.5*2.55-.5*1.3,-.5*1.3) -- (-.5*1.3,-.5*1.3);
			\draw[big arrow] (.5,.6) -- (1.72,0) {};
			\draw[big arrow,yshift=-30] (.5,-.6) -- (1.72,0) {};
			\draw[thick] (C0') -- (C13) -- (C2);
		\end{tikzpicture}
				\end{array}\\\hline
				\text{Matter representation} &\text{adjoint+vector}:  (\mathbf{10}\oplus \mathbf{5}),\quad n_{\bf{10}}=1+K^2, \quad n_{\bf{5}}=3 K^2\\\hline 
				\text{Euler characteristic} &  \frac{ 4 L ( 3 + 4 L ) }{(1+2L)^2} c(TB)\\\hline
				\text{Triple intersections} &6 \mathcal F= 8 L^2 (-\alpha_0^3 + 3  \alpha_0 \alpha_1^2 - 4 \alpha_1^3 + 3  \alpha_1^2 \alpha_2 -  \alpha_2^3)\\\hline
%
	\multicolumn{2}{c}{		\begin{matrix}\\	\text{\textbf{SO($6$)-model}}   \end{matrix}  } \\\hline
				\text{Weierstrass equation}  &y^2 z + a_1 x y z= x^3 + mt x^2 z + s^2 x z^2\\\hline
			\text{Discriminant} &\Delta =s^4 ((a_1^2 + 4m s)^2 - 64 s^2 )\\\hline
				\text{Singular fibers}&\begin{array}{c}	\begin{tikzpicture}[scale=1.1]
			\node[draw,circle,thick,scale=.9](C0) at (0,0) {};
			\node[draw,circle,thick,scale=.9](C1) at (-.5,-.5) {};
			\node[draw,circle,thick,scale=.9](C2) at (.5,-.5) {};
			\node[draw,circle,thick,scale=.9](C3) at (0,-1) {};
			\draw[thick] (C0) -- (C2) -- (C3) -- (C1) -- (C0);
			\draw[big arrow] (1,-.5) -- node[above,midway]{{}} (1.6,-.5);
			\node[draw,circle,xshift=80,thick,scale=.9](C0') at (0,0) {};
			\node[draw,circle,xshift=80,thick,scale=.65,scale=.7](C13) at (0,-.5) {\LARGE 2};;
			\node[draw,circle,xshift=80,thick,scale=.9](C2) at (-.5,-.5) {};
			\node[draw,circle,xshift=80,thick,scale=.9](C3-) at (.5,-.5) {};
			\node[draw,circle,xshift=80,thick,scale=.9](C3+) at (0,-1) {};
			\draw[thick] (C0') -- (C13) -- (C3-);
			\draw[thick] (C2) -- (C13)--(C3+);
						\draw[yshift=85,scale=.45,domain=-1.2:1.2,variable=\x, thick,xshift=10] plot({\x*\x-1,\x*\x*\x-\x-5});
			\draw[yshift=15,scale=.45,domain=-1.2:1.2,variable=\x, thick,xshift=10] plot({\x*\x-1,\x*\x*\x-\x-5});
			\draw[big arrow] (.5,.6) -- (1.60,0) {};
			\draw[big arrow,yshift=-30] (.5,-.6) -- (1.60,0) {};
		\end{tikzpicture}\end{array}\\\hline
		\text{Matter representation} & \text{adjoint+vector:}\   (\mathbf{15}\oplus \mathbf{6}),\quad 
		\quad n_{\bf{15}}=1+K^2, \quad n_{\bf{6}}=2 K^2 \\\hline
				\text{Euler characteristic} & \frac{12 L }{1 + 2 L } c(TB)\\\hline
				\text{Triple intersections} &
		  \begin{array}{c}6 \mathcal F^- = 2 L^2(-4 \alpha_0^3 + 3 \alpha_0 \alpha_1^2 + 3 \alpha_0 \alpha_3^2 + 6 \alpha_0 \alpha_1 \alpha_3)\\
		+ 2 L^2 ( -4 \alpha_1^3 - 4 \alpha_2^3  - 6 \alpha_3^3 - 6 \alpha_1^2 \alpha_3+ 3 \alpha_1^2 \alpha_2 + 3 \alpha_3^2 \alpha_2 + 6 \alpha_1 \alpha_2 \alpha_3 ),\\
		6 \mathcal F^+ =6 \mathcal F^- +4 L^2 (\alpha_3 - \alpha_1)^3 \end{array}\\\hline
			\end{array} 
		\end{array}
	$\\
	\\
	\hline
	\end{tabular}
	}
	\end{center}
	\caption{Summary of results for the geometry of the SO(3), SO(5) and SO(6)-models.}
	\label{table:summary}
\end{table}
\clearpage

\newpage

\section{Preliminaries and summary of results} \label{Sec:preliminaries}

\subsection{Weierstrass models}\label{Sec:Conv}
In this section, we introduce our conventions and some basic  definitions. 
We mostly follow the presentation of \cite{G2} with some basic adaptations to the case of a simple, connected compact group $G$.

\begin{defn}[Weierstrass model]
Consider a variety $B$ endowed with a line bundle $\mathscr{L}\rightarrow B$.
A Weierstrass model $Y_0 \rightarrow B$ over $B$ is a hypersurface  cut out by the zero locus of a section of the line bundle of $\mathscr{O}(3)\otimes\pi^* \mathscr{L}^{\otimes 6}$ in the projective bundle  $\mathbb{P}(\mathscr{O}_B \oplus\mathscr{L}^{\otimes 2} \oplus\mathscr{L}^{\otimes 3})\rightarrow B$. 
We denote by $\mathscr{O}(1)$  the dual of the tautological line bundle of the projective bundle, and denote by $\mathscr{O}(n)$ ($n>0$) its $n$th-tensor product $\mathscr O(1)^{\otimes n}$.  
The relative projective coordinates of the $\mathbb{P}^2$ bundle are denoted by $[x:y:z]$. In particular,  $x$ is a section of $\mathscr{O}(1)\otimes\pi^* \mathscr{L}^{\otimes 2}$, $y$ is a section of $\mathscr{O}(1)\otimes\pi^* \mathscr{L}^{\otimes 3}$, and $z$ is a section of $\mathscr{O}(1)$. 
 Following Tate and Deligne's notation, the defining equation of a  Weierstrass model is
$$
Y_0: \quad  zy(y +a_1 x+ a_3 z) -(x^3 +a_2 x^2 z+a_4 xz^2 +a_6 z^3)=0,
$$
 where the coefficient $a_i$ ($i=1,2,3,4,6$) is a section of $\mathscr{L}^{\otimes i}$ on $B$. Such a hypersurface is an elliptic fibration since over the generic point of the base, the fiber is a nonsingular  cubic planar curve with a rational point ($x=z=0$). 
    We use the convention of  Deligne's formulaire \cite{Deligne.Formulaire} and introduce the following definitions: 
\begin{equation}\label{Eq:Formulaire}
\begin{aligned}
b_2 &= a_1^2+ 4 a_2,\quad
b_4 = a_1 a_3 + 2 a_4 ,\quad
b_6  = a_3^2 + 4 a_6 , \quad
b_8  =b_2 a_6 -a_1 a_3 a_4 + a_2 a_3^2-a_4^2,\\
c_4 & = b_2^2 -24 b_4, \quad
c_6 = -b_2^3+ 36 b_2 b_4 -216 b_6,\\
\Delta &= -b_2^2 b_8 -8 b_4^3 -27 b_6^2 + 9 b_2 b_4 b_6,\quad  j =\frac{c_4^3}{\Delta}.
\end{aligned}
\end{equation}
The above quantities are subject to the  relations
$$
4 b_8 =b_2 b_6 -b_4^2,~~ 1728 \Delta=c_4^3 -c_6^2. 
$$ 
\end{defn}

 The discriminant locus is the subvariety of $B$ cut out by the equation $\Delta=0$, and is the locus of points  $p$ of the base $B$ such that the fiber over $p$ (i.e. $Y_{0}|_p$)  is singular. 
 Over a generic point of $\Delta$, the fiber is a nodal cubic that  degenerates to a cuspidal  cubic over the codimension two locus $c_4=c_6=0$. Up to isomorphism, the
 $j$-invariant $j=c_4^3/ \Delta$ uniquely characterizes nonsingular elliptic curves.

\subsection{Intersection theory and blowups} \label{sec:intersectionBL}
The intersection theory discussed in this paper essentially relies on  three theorems from \cite{Euler}, where we
use the fact that each crepant resolution  is expressed by a sequence of  blowups. 
The first theorem, due to Aluffi, expresses the Chern class of a blowup along a local complete intersection. The second theorem describes the pushforward associated to a blowup whose center is a local complete intersection. The third and final theorem  provides an explicit map of an  analytic expression in the Chow ring of a projective bundle to  the Chow ring of its base.

\begin{thm}[Aluffi, {\cite[Lemma 1.3]{Aluffi_CBU}}]
\label{Thm:AluffiCBU}
Let $Z\subset X$ be the  complete intersection  of $d$ nonsingular hypersurfaces $Z_1$, \ldots, $Z_d$ meeting transversally in $X$.  
Let  $f: \widetilde{X}\longrightarrow X$ be the blowup of $X$ centered at $Z$. 
We denote the exceptional divisor of $f$  by $E$. The total Chern class of $\widetilde{X}$ is then:
\begin{equation}
c( T{\widetilde{X}})=(1+E) \left(\prod_{i=1}^d  \frac{1+f^* Z_i-E}{1+ f^* Z_i}\right)  f^* c(TX).
\end{equation}
\end{thm}

\begin{thm}[Esole--Jefferson--Kang,  see  {\cite{Euler}}] \label{Thm:Push}
Let the nonsingular variety $Z\subset X$ be a complete intersection of $d$ nonsingular hypersurfaces $Z_1$, \ldots, $Z_d$ meeting transversally in $X$. Let $E$ be the class of the exceptional divisor of the blowup $f:\widetilde{X}\longrightarrow X$ centered at $Z$.
 Let $\widetilde{Q}(t)=\sum_a f^* Q_a t^a$ be a formal power series with $Q_a\in A_*(X)$.
 We define the associated formal power series  ${Q}(t)=\sum_a Q_a t^a$, whose coefficients pullback to the coefficients of $\widetilde{Q}(t)$. Then the pushforward $f_*\widetilde{Q}(E)$ is
\begin{equation*}
f_*  \widetilde{Q}(E) =  \sum_{\ell=1}^d {Q}(Z_\ell) M_\ell, \quad \text{where} \quad  M_\ell=\prod_{\substack{m=1\\
 m\neq \ell}}^d  \frac{Z_m}{ Z_m-Z_\ell }.
\end{equation*}
\end{thm}

\begin{thm}[Esole--Jefferson--Kang, see  \cite{Euler}]\label{Thm:PushH}
Let $\mathscr{L}$ be a line bundle over a variety $B$ and $\pi: X_0=\mathbb{P}[\mathscr{O}_B\oplus\mathscr{L}^{\otimes 2} \oplus \mathscr{L}^{\otimes 3}]\longrightarrow B$ a projective bundle over $B$. 
 Let $\widetilde{Q}(t)=\sum_a \pi^* Q_a t^a$ be a formal power series in  $t$ such that $Q_a\in A_*(B)$. Define the auxiliary power series $Q(t)=\sum_a Q_a t^a$. 
Then 
\begin{equation*}
\pi_* \widetilde{Q}(H)=-2\left. \frac{{Q}(H)}{H^2}\right|_{H=-2L}+3\left. \frac{{Q}(H)}{H^2}\right|_{H=-3L}  +\frac{Q(0)}{6 L^2},
\end{equation*}
 where  $L=c_1(\mathscr{L})$ and $H=c_1(\mathscr{O}_{X_0}(1))$ is the first Chern class of the dual of the tautological line bundle of  $ \pi:X_0=\mathbb{P}(\mathscr{O}_B \oplus\mathscr{L}^{\otimes 2} \oplus\mathscr{L}^{\otimes 3})\rightarrow B$.
\end{thm}

\begin{notation}[Blowups]\label{Notation:Blowups}
 Let $X$ be a nonsingular variety. 
 Let $Z\subset X$ be a complete intersection defined by the transverse intersection of $r$ hypersurfaces $Z_i=V(g_i)$, where  $(g_1, \cdots, g_r)$ is a regular sequence. 
 We denote the blowup of a nonsingular variety $X$ with center  the complete intersection $Z$ by 
 $$\begin{tikzpicture}
	\node(X0) at (0,-.3){$X$};
	\node(X1) at (3,-.3){$\widetilde{X}.$};
	\draw[big arrow] (X1) -- node[above,midway]{$(g_1,\cdots ,g_{r}|e_1)$} (X0);	
	
	\end{tikzpicture}
	$$
The exceptional divisor is $E_1=V(e_1)$.	
 We abuse notation and use the same symbols for $x$, $y$, $s$, $e_i$ and their successive proper transforms. We do not write the obvious pullbacks. 
\end{notation}

\begin{defn}[Resolution of singularities]
A resolution of singularities of a variety $Y$ is a proper birational morphism $\varphi:\widetilde{Y}\longrightarrow Y$  such that  
$\widetilde{Y}$ is nonsingular
and  $\varphi$ is an isomorphism away  from the singular  locus of $Y$. In other words, $\widetilde{Y}$ is nonsingular and  if $U$ is the singular locus of $Y$, $\varphi$ maps $\varphi^{-1}(Y\setminus U)$ isomorphically  onto $Y\setminus U$.  
\end{defn}

\begin{defn}[Crepant birational map]
A  birational map $\varphi:\widetilde{Y}\to Y$ between two algebraic varieties with  $\mathbb{Q}$-Cartier canonical classes is said to be {\em crepant} if it preserves the canonical class, i.e.  
$
K_{\widetilde{Y}}=\varphi^* K_Y.
$
\end{defn}

\subsection{Hyperplane arrangement} \label{sec:hyper}
 Following a common convention in physics, we denote an  irreducible representation \textbf{R} of a Lie algebra $\mathfrak{g}$ by its dimension (in boldface.) 
The weights are denoted by $\varpi^I_j$ where the upper index I denotes the representation $\textbf{R}_I$ and the lower index $j$ denotes a particular  weight of the representation $\textbf{R}_I$. 
 Let $\phi$ be a vector of the coroot space of $\mathfrak{g}$   in the basis of the fundamental coroots.  
 Each weight $\varpi$ defines a linear form $(\phi, \varpi)$ by the natural evaluation on a coroot. We recall that fundamental coroots are dual to fundamental weights. Hence, with our choice of conventions, 
 $(\phi, \varpi)$ is the usual Euclidian scalar product.

\begin{defn}[Hyperplane arrangement {I$(\mathfrak{g},\mathbf{R})$}]
The hyperplane arrangement I$(\mathfrak{g},\mathbf{R})$ is defined inside the dual fundamental Weyl chamber of $\mathfrak{g}$ and its hyperplanes are the kernel of the weights of the representation $\mathbf{R}$ \cite{Hayashi:2014kca,EJJN1,EJJN2}.  
\end{defn}
A crepant resolution over a Weierstrass model is always a minimal model over the Weierstrass model. Distinct minimal models are connected by a sequence of flops. 
The geography of flops (or equivalently of its extended K\"ahler-cone) of a 
 $G$-model  is conjectured to be given by the chamber structure of the hyperplane arrangement I$(\mathfrak{g},\mathbf{R})$, where $\mathfrak{g}$ is the Lie algebra whose type is dual to the dual graph defined by the fibral divisors of the $G$-model and $\mathbf{R}$ is the representation whose weights are given by minus the intersection numbers of rational curves comprising singular fibers located over codimension-two points in the base \cite{IMS}. 

\subsection{$G$-models and crepant resolutions of Weierstrass models}
One of the  fundamental  insights of F-theory 
is that the geometry of  an elliptic fibration $\varphi:Y\to B$ naturally determines a triple $(\mathbf{R}, G,\mathfrak{g})$ where $\mathbf{R}$ is a representation of a reductive Lie group $G$ with a semi-simple Lie algebra $\mathfrak{g}$. Moreover, the fundamental group  of the Lie group $G$ is isomorphic to the Mordell--Weil group of the elliptic fibration.
We call such an elliptic fibration a $G$-model.  $G$-models are  used in M-theory and F-theory compactifications on elliptically fibered varieties to geometrically engineer gauge theories with gauge group $G$ and matter transforming in the representation $\mathbf{R}$ of the gauge group. Under mild assumptions, a $G$-model is  birational to a singular Weierstrass model \cite{Deligne.Formulaire,MumfordSuominen,Nakayama.OWM}. Starting from a singular Weierstrass model, we can retrieve a smooth elliptic fibration via a resolution of singularities. 
In the best situation, we can ask the resolution to be crepant. 
When crepant resolutions are not possible, we settle for a partial resolution with $\mathbb{Q}$-factorial terminal singularities. 
In both cases, the (partial)-resolution is a minimal model over the Weierstrass model in the sense of Mori's theory. In the physically relevant cases, we are mostly interested in crepant resolutions.
Starting from threefolds, minimal models are not unique but connected by a finite sequence of flops. 
The network of flops of a given $G$-model is described by a  hyperplane arrangement I($\mathfrak{g},\mathbf{R}$) defined in terms of the Lie algebra $\mathfrak{g}$ and the representation $\mathbf{R}$. 
Two chambers of a hyperplane arrangement sharing a common boundary wall of one dimension lower are said to be adjacent. 
  The adjacent graph of the chambers of the hyperplane arrangement I($\mathfrak{g},\mathbf{R}$) of a $G$-model with a triple $(\mathbf{R},G,\mathfrak{g})$
    is conjectured to be isomorphic to the network of flops of the minimal models derived by crepant resolutions of the Weierstrass model birational to the $G$-model. 
 Each of the chambers also coincides with a Coulomb phase of a 5D ${\mathcal N}=1$ supergravity theory with gauge group $G$ and hypermultiplets in the representation \textbf{R}, related to compactification of M-theory on the $G$-model assuming that the $G$-model is a Calabi-Yau threefold. 

 \subsection{Understanding $G$-models for orthogonal groups of small rank}
 Several $G$-models have been studied from the perspective of M-theory and F-theory compactifications.
The SU($2$), SU($3$), SU($4$), and SU($5$)-models are studied in \cite{ESY1,ESY2,EY}.  
Aspects of the general SU($n$)-models are discussed in \cite{ES}. The 
 G$_2$,  Spin($7$), and Spin($8$)-models are  studied  in \cite{G2}, the F$_4$-models in \cite{F4}, the E$_7$-models in \cite{E7}. 
 Semi-simple Lie groups of rank two or three have been studied recently: Spin($4$)=SU($2$)$\times$ SU($2$) and  SO($4$) in \cite{SO4}, 
 SU($2$)$\times$ G$_2$ in \cite{SU2G2},    SU($2$)$\times$ SU($3$) in \cite{SU2SU3}, and SU($2$)$\times$ SU($4$) in \cite{Esole:2017hlw}, 
The Euler characteristics of $G$-models defined by crepant resolutions of Weierstrass models resulting from Tate's algorithm have been computed recently in \cite{Euler}  for $G$ a connected and simple group for characteristic numbers of other type of ellliptic fibrations \cite{ Esole:2017hlw, Esole:2012tf}.  
The Hodge numbers for $G$-models that are Calabi-Yau threefolds are also computed in \cite{Euler}. 
For connected compact simple Lie groups, there are few cases left to study, with one subtle case being that of special orthogonal groups of small rank.

 Special orthogonal groups SO($7+2n$)  or SO($8+2n$) require an elliptic fibration with a Mordell--Weil group $\mathbb{Z}/2\mathbb{Z}$ and a discriminant locus containing an irreducible  component  $S$ such that the generic fiber over $S$ is respectively of Kodaira type  I$_{n}^{*\text{ns}}$ and  I$_{n}^{*\text{s}}$ while the generic fibers over other irreducible components of the discriminant locus are irreducible curves. 
In particular, SO($m$)-models with small rank  ($3\leq m\leq 6$) cannot be achieved with fibers of type I$_n^*$ as $n$ would have to be negative.

   The aim of this paper is to study the geometry of  SO($n$)-models in the case of $n= 3,5,6$:
$$\text{SO}(3), \quad \text{SO}(5), \quad\text{or}\quad  \text{SO}(6).$$
The SO($4$)-model  is studied in \cite{SO4}, and is  more complicated as SO($4$) is not a simple group  and thus requires a collision of singularities. Moreover, the dual graph $\widetilde{\text{A}}_1$ can be realized by several Kodaira types, which implies that there are several distinct ways to construct an SO($4$)-model \cite{SO4}. 

 The geometric engineering of SO($3$), SO($5$), and SO($6$)-models relies on the accidental isomorphisms 
$$
\mathfrak{so}_3 \cong \mathfrak{su}_2\quad 
\mathfrak{so}_5 \cong \mathfrak{sp}_4\quad 
\mathfrak{so}_5 \cong \mathfrak{su}_4. 
$$
The SO($3$), SO($5$), and SO($6$)-models  are respectively derived by imposing the existence of a $\mathbb{Z}/2\mathbb{Z}$ Mordell--Weil group in the SU($2$), Sp($4$), and SU($4$)-models. The corresponding Kodaira fibers are respectively 
$$
\text{I}_2 \  \text{or}  \   \text{III},\quad \text{I}_4^{\text{ns}},\quad \text{and} \quad \text{I}_4^{\text{s}},
$$
and their  dual graphs are respectively the affine Dynkin diagrams (see Figure \ref{fig:fibergraphs})
 $$
\widetilde{\text{A}}_{1},\quad   \widetilde{\text{C}}_2^t,\quad  \text{and}\quad   \widetilde{\text{A}}_{3}.
 $$

    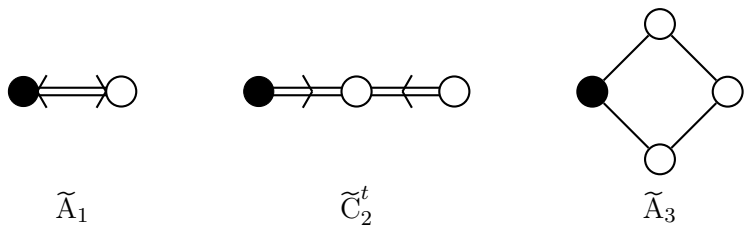
\begin{figure}[htb]
	\begin{center}
		$
			\begin{array}{cccccccc}
			\begin{array}{c}
			\begin{tikzpicture}
				\draw[fill=black]	
					(0,0) circle [radius=.2];
				\draw[thick]
					(1.3,0) circle [radius=.2] node {};
				\draw[thick] 
					(0,-.05) --++ (1.1,0)
					(0,+.05) --++ (1.1,0);                      
				\draw[thick,xshift=-9]
					(.5,0) --++ (60:.25)
					(.5,0) --++ (-60:.25);
				\draw[thick,xshift=17]
					(.5,0) --++ (120:.25)
					(.5,0) --++ (-120:.25);
			\end{tikzpicture}
			\end{array}
			&&&
			\begin{array}{c}
			\begin{tikzpicture}
				\draw[fill=black]	
					(0,0) circle [radius=.2];
				\draw[thick]
					(1.3,0) circle [radius=.2] node{};
				\draw[thick]
					(2.6,0) circle [radius=.2] node{};
				\draw[thick] 
					(0,-.05) --++ (1.1,0)
					(0,+.05) --++ (1.1,0); 
				\draw[thick]          
					(1.5,-.05) --++ (.9,0) 
					(1.5,+.05) --++ (.9,0);        
				\draw[thick,xshift=6]
					(.5,0) --++ (120:.25)
					(.5,0) --++ (-120:.25);
				\draw[thick,xshift=6]
					(1.7,0) --++ (60:.25)
					(1.7,0) --++ (-60:.25);
			\end{tikzpicture}
			\end{array}
			&&&
			\begin{array}{c}
			\begin{tikzpicture}[]
				\node[draw,circle,thick,scale=1.1] (u) at (0,0){};
				\node[draw,circle,thick,fill=black,scale=1.1] (l) at (-.9,-.9){};
				\node[draw,circle,thick,scale=1.1] (d) at (0,-1.8){};
				\node[draw,circle,thick,scale=1.1] (r) at (.9,-.9){};
				\draw[thick] (u) to (l) to (d) to (r) to (u);
			\end{tikzpicture}
			\\ 
			\end{array}\\
			\widetilde{\text{A}}_1   & &  & \widetilde{\text{C}}^t_2 &  &&\widetilde{\text{A}}_3  & 
			\end{array}
		$
	\end{center}
	\caption{Affine Dynkin diagrams corresponding to the respective Lie groups SO(3), SO(5), and SO(6). In the above graphs, the black node represents the affine node; deleting this node produces the corresponding finite Dynkin diagrams of type A$_1 \cong $ B$_1$, C$_2 \cong $ B$_2$ and A$_3 \cong $ D$_3$.}
	\label{fig:fibergraphs}
\end{figure}

\subsection{Spin($n$)-models versus SO($n$)-models}

By applying Tate's algorithm in reverse one can construct a $G$-model for any of the following simply connected and simple Lie groups:  
$$\text{SU}(n)\ (n\geq 2),\quad \text{Sp}(2n)\ (n\geq 2), \quad \text{Spin}(n)\  (n\geq 7), \quad \text{G}_2,\quad \text{F}_4, \quad \text{E}_6, \quad \text{E}_7, \quad \text{E}_8$$ 
Spin groups of low rank are retrieved by the following  accidental isomorphisms:
$$
\begin{aligned}
\text{Spin}(3) \cong \text{SU}(2), \quad
 \text{Spin}(4) &\cong SU(2)\times \text{SU}(2),\quad
  \text{Spin}(5)\cong \text{USp}(4), \quad
   \text{Spin}(6) \cong \text{SU}(4).
\end{aligned}
$$
On the other hand, orthogonal groups are much more complicated to handle than Spin groups since they require the Mordell--Weil group to be isomorphic to  $\mathbb{Z}/2\mathbb{Z}$. 
 For $n\geq 3$, Spin($n$) is a simply-connected double-cover of SO($n$):
$$
1\longrightarrow  \mathbb{Z}/2\mathbb{Z} \longrightarrow  \text{Spin($n$)} \overset{\pi}{\longrightarrow}  \text{SO($n$)} \longrightarrow 1.
$$
It follows that an SO($n$)-model  with $n\geq 7$ is given by an elliptic fibration 
 with a  Mordell--Weil  group $\mathbb{Z}/2\mathbb{Z}$ and a  discriminant locus  which contains an  irreducible component $S$ such that the fiber over the generic point of $S$ is of  type I$_k^*$ and the fibers over the remaining generic points of $\Delta$ are irreducible curves (type I$_1$ or II).
More specifically, a fiber of type  I$_{k\geq 0}^{*\text{s}}$ with a Mordell--Weil torsion $\mathbb{Z}/2\mathbb{Z}$ gives an SO($8+2k$)-model while  
a fiber of type  I$_{k\geq 1}^{*\text{ns}}$ gives an SO($7+2k$)-model. An SO($7$)-model requires a fiber of type I$_0^{* \text{ss}}$ and Mordell--Weil torsion $\mathbb{Z}/2\mathbb{Z}$. 
These points are summarized in Table \ref{Table:SO}.

\begin{table}[htb]
\begin{center}
\begin{tabular}{ |c |c |c| }
\hline
$G$ & Generic fiber over $S$ & Mordell--Weil group \\
\hline
Spin($7$) &  I$_{0}^{*\text{ss}}$  & trivial\\
Spin($8+2k$) &  I$_k^{*\text{s}}$ & trivial\\
Spin($9+2k$) &  I$_{1+k}^{*\text{ns}}$  & trivial\\
\hline 
SO($7$) & I$_0^{*ss}$\ &  $\mathbb{Z}/2\mathbb{Z}$\\
SO($8+2k$) & I$_k^{*s}$\ &  $\mathbb{Z}/2\mathbb{Z}$\\
SO($9+2k$) & I$_{1+k}^{*\text{ns}}$ & $\mathbb{Z}/2\mathbb{Z}$  \\
\hline
\end{tabular}
\end{center}
\caption{   
Spin($n$)-models vs SO($n$)-models. 
\label{Table:SO}}
\end{table}

\subsection{Weierstrass models with Mordell--Weil group $\mathbb{Z}/2\mathbb{Z}$ }
A genus-one fibration  $\varphi:Y\to B$ is a proper surjective morphism $\varphi$ between two algebraic varieties $Y$ and $B$ such that the generic fiber  is a smooth projective curve of genus one. 
 A genus-one fibration  $\varphi:Y\to B$ is said to be an elliptic fibration if  $\varphi$  is  endowed with a  choice of rational section.  A rational section of a morphism $\varphi:Y\to B$ is morphism  $\sigma: U\to Y$ such that  $U$ a Zariski open subset  of $B$ and   $\varphi\circ \sigma$ restricts to the identity on $U$. 
The rational section turns the generic fiber into a bona fide elliptic curve: a genus-one curve with a choice of a rational point. 
Given an elliptic curve $E$ defined over a field $k$, the Mordell--Weil group of $E$ is the group of $k$-rational points of $E$.  
The Mordell--Weil theorem states that the Mordell--Weil group is an abelian group of finite rank. This is a theorem proven by Mordell in the case of an elliptic curve defined over the rational numbers $\mathbb{Q}$ and later generalized to number fields and abelian varieties by Weil. For an elliptic fibration, the Mordell--Weil group is the Mordell--Weil of its generic fiber or equivalently, the group of its rational sections. 
Throughout this paper, we work over the complex numbers. We write $V(f_1, \cdots, f_n)$ for the algebraic set  defined by the solutions of $f_1=\cdots =f_n=0$.

An important model of an  elliptic fibration is the Weierstrass model, which expresses an elliptic fibration as a certain hypersurface in a $\mathbb{P}^2$-projective bundle with the generic fibers being a planar cubic curve in the $\mathbb{P}^2$-fiber. 
Weierstrass models are very convenient as they have ready-to-use formulas which compute important data such as the  $j$-invariant and the discriminant locus. 
Tate's algorithm provides a simple procedure to identify the type of  a singular fiber by manipulating the coefficients of the  Weierstrass model  \cite{Papadopoulos,Tate.Alg}. 
Intersection theory on an elliptic fibration is also made much easier in a Weierstrass model as the ambient space is a projective bundle. 
On a Weierstrass model, the Mordell--Weil group law can be expressed geometrically by the chord-tangent algorithm.

   Weierstrass models are written in the notation of Deligne  and Tate \cite{Deligne.Formulaire} reviewed in Section  \ref{Sec:Conv}. 
   The Weierstrass model is expressed as the locus
$$
zy^2+ a_1 xy z+ a_3 xy z =x^3+ a_2 x^2z + a_4 xz^2 + a_6z^3,
$$ 
in a projective $\mathbb{P}^2$-bundle with projective coordinates $[x:y:z]$, as is discussed in more detail in Section \ref{Sec:Conv}. 
The fundamental line bundle of a Weierstrass model is denoted by $\mathscr{L}$ and the coefficient $a_i$ ($i=1,2,3,5,6$) is a section of $\mathscr{L}^{\otimes i}$. 
If the Weierstrass model is Calabi-Yau, its canonical class is trivial and  the line bundle $\mathscr{L}$ is equal to the anti-canonical line bundle of the base $B$ of the fibration. 
Since we work in characteristic zero, we can shift $y$ to eliminate $a_1$ and $a_3$ and write the Weierstrass model as follows. 
$$
zy^2 =x^3+ a_2 x^2z + a_4 xz^2 + a_6z^3.
$$
The zero section is $x=z=0$ and the involution defined by sending a rational point $[x_0:y_0:z_0]$ to its opposite with respect to the Mordell--Weil group takes the following simple form:
$$
[x_0:y_0:z_0]\mapsto [x_0:-y_0:z_0].
$$
In particular, if $a_6$ is identically zero, there is a new rational  section $x=y=0$, which is  its own inverse and it follows that the Mordell--Weil group is generically $\mathbb{Z}/2\mathbb{Z}$.
An elliptic fibration with Mordell--Weil group $\mathbb{Z}/2\mathbb{Z}$ can always be put in the form 
$$
zy^2 =x(x^2+ a_2 xz + a_4 z^2).
$$
The discriminant locus of this elliptic fibration is 
\begin{equation}\nonumber
\Delta= 16 a_4^2 (4a_4 -a_2^2).
\end{equation}
One can check that the fiber over the locus $V(a_4)$ is of Kodaira type I$_2$ while the fiber over the other component of the reduced discriminant is $I_1$. 
It follows that the generic elliptic fibration with a $\mathbb{Z}/2\mathbb{Z}$ Mordell--Weil group is an SO($3$)-model and the class of the divisor supporting the group SO($3$) is necessarily a section of the line bundle $\mathscr{L}^{\otimes 4}$.

A generic  Weierstrass model with a Mordell--Weil torsion subgroup $\mathbb{Z}/2\mathbb{Z}$ is given by the following theorem which is a direct consequence of a classic result in the study of elliptic curves in number theory (see for example \cite[\S5 of Chap 4]{Husemoller}) and was first discussed in a string theoretic setting by Aspinwall and Morrison  \cite{Aspinwall:1998xj}. 
\begin{thm}\label{Thm:WZ2}
An elliptic fibration over a smooth base $B$ and with  Mordell--Weil group $\mathbb{Z}/2\mathbb{Z}$ is birational to the  following (singular) Weierstrass model.
\begin{equation} \nonumber
zy^2= x(x^2 + a_2 x z+ a_4z^2),
\end{equation}
\end{thm}
The section $x=y=0$ is the generator of the $\mathbb{Z}/2\mathbb{Z}$ Mordell--Weil group  and $x=z=0$ is the neutral element of the  Mordell--Weil group. The discriminant of this Weierstrass model is
\begin{equation}\nonumber
\Delta=16 a_4^2(a_2^2 -4 a_4).
\end{equation}
This model has a fiber of type I$_2^{\text{ns}}$ over $V(a_4)$, and fibers of type I$_1$ over $V(a_2^2-4a_4)$. 
As we will see after a crepant resolution of singularities, at the collision of these two components, namely at $V(a_2,a_4)$, we get a fiber of type III. 
The dual graph of I$_2^{\text{ns}}$ is the affine Dynkin diagram $\widetilde{\text{A}}_1$. It follows that the Lie algebra $\mathfrak{g}$ associated with this elliptic fibration is A$_1$, the Lie algebra of the simply connected group SU($2$). 
The generic Weierstrass model with a $\mathbb{Z}/2\mathbb{Z}$ Mordell--Weil group automatically gives an SO($3$)-model since we have the exact sequence 
$$
1\longrightarrow  \mathbb{Z}/2\mathbb{Z} \longrightarrow  \text{SU($2$)} \longrightarrow  \text{SO($3$)} \longrightarrow 1.
$$
The SO($3$) and the SU($2$)-models are rather different. For example, an SO($3$)-model cannot have matter charged in a spin representation of A$_1$  while an SU($2$)-model can always have matter in the representation $\mathbf{2}$ of SU($2$) \cite{ESY1}, which is a spin representation. 
Generically, the SO($3$)-model only has matter in the adjoint representation while the SU($2$)-model has matter in both the adjoint and  fundamental representations. 
Geometrically, this is due to the existence of fibers of type I$_3$ in codimension two for SU($2$)-model; by contrast, the SO($3$)-model only has fibers of type III over codimension-two points.

\subsection{Canonical forms for SO($3$), SO($4$), SO($5$), and SO($6$)-models}

The SO($3$)-model is the generic case of a Weierstrass model with Mordell--Weil group $\mathbb{Z}/2\mathbb{Z}$  as given in Theorem \ref{Thm:WZ2} and has been studied in \cite{Aspinwall:1998xj,Mayrhofer:2014opa}. 
The SO($4$)-model was constructed in \cite{SO4}. 
 To the authors' knowledge, the SO($5$) and SO($6$)-models have not been constructed explicitly before and were announced in \cite{Euler} where we computed their Euler characteristics  \cite{Euler}  and additional characteristic numbers were 
computed in \cite{Char1}.
The SO($5$)-model is obtained from the SO(3)-model  via a base change that converts the section $a_4$ to a perfect square $a_4=s^2$ of an irreducible and reduced Cartier divisor so that the fiber over $V(s)$ is of type I$_4^{\text{ns}}$. 
If instead we had $a_4=a s^2$, the group would be  semi-simple. 
The SO(6)-model is subsequently derived from the SO(5)-model by requiring that the section $a_2$ is a perfect square modulo $s$ so that the fiber type over $S=V(s)$ is of type I$_4^{\text{s}}$.

The Weierstrass models for the SO($3$), SO($4$), SO($5$), and SO($6$)-models are as follows:
 \begin{align}
\label{eqn:SO3def} \text{SO}(3)\text{-model}: \quad & zy^2= x(x^2 + a_2 x z+ a_4z^2), \\
\text{SO}(4)\text{-model}: \quad & zy^2= x(x^2 + a_2 x z+ s t z^2), \\
\label{eqn:SO5def} \text{SO}(5)\text{-model}: \quad & zy^2= x(x^2 + a_2 x z+ s^2 z^2),\\
\text{SO}(6)\text{-model}: \quad & zy^2+ a_1 y x  z= x(x^2 + m s  x z+ s^2z^2), ~~~ m\neq 0, \pm 2.
\end{align}

The SO($4$)-model is derived from the SO($3$)-model by a base change $a_4\to s t$ where $S=V(s)$ and $T=V(t)$. 
The SO($5$)-model is derived from the SO($3$)-model by a base change $a_4\to s^2$ where $S=V(a_4)$ for SO($3$) and $S=V(s)$ for SO($5$)-models. 
The SO($6$)-model is derived from the SO($5$)-model by forcing the  fiber over the generic point of $S$ to be type I$_4^{\text{s}}$ rather than type I$_4^{\text{ns}}$. 
 For SO($5$) and SO($6$), $a_2$ cannot be identically zero, otherwise the Mordell--Weil group becomes $\mathbb{Z}/2\mathbb{Z}\times \mathbb{Z}/2\mathbb{Z}$. 
 Moreover, in the SO($5$)-model, if $a_2$ is identically zero,  the fiber over the generic point $S=V(s)$ becomes I$_0^{*\text{s}}$ and the Lie algebra becomes D$_4$.
In the SO($6$)-model, we can complete the square in $y$ and end up with 
 $ zy^2= x(x^2 + a_2 x z+ s^2 z^2)$ where $a_2=a_1^2/4+ m s$. This shows that the SO($6$)-model is a limiting case of the SO($5$)-model in which $a_2$ is a perfect square modulo $s$, hence the fiber  over $S$ is of type I$_4^{\text{s}}$ instead of 
  I$_4^{\text{ns}}$.

   $\text{SO}(4)\cong \text{SO}(3) \times \text{SO}(3)$ is the product of two simple groups and therefore requires two irreducible components $S=V(s)$ and $T=V(t)$, which also implies that their product $st$ is a section of $\mathscr{L}^{\otimes 6}$ \cite{SO4}.   
For the SO($3$), SO($5$), and SO($6$)-models, the class of the divisor $S$ supporting the gauge group is completely fixed in terms of the fundamental line bundle $\mathscr{L}$ of the Weierstrass model.\footnote{ 
The same is true for $G=$SU($2n$)/($\mathbb{Z}/2\mathbb{Z}$) and $G=$Sp($4+2n$)/($\mathbb{Z}/2\mathbb{Z}$) where Sp($2m$) is the  compact connected Lie group with Lie algebra C$_m$.
In particular, it follows that in the Calabi-Yau case,  the class of the divisor $S$ depends only on the canonical class of the base. }
The divisor  $S$ is  a section of  $\mathscr{L}^{\otimes 4}$ for SO($3$)-models and of  $\mathscr{L}^{\otimes 2}$ for SO($5$) and SO($6$)-models. 
If we denote by $L=c_1(\mathscr{L})$ the first Chern class of $\mathscr{L}$, then the class of $S$ in the Chow ring of the base  is respectively $4L$ for SO($3$)-models; and $2L$ for SO($5$) and SO($6$)-models. In the Calabi-Yau case, $L=-K$ where $K$ is the canonical class of the base. 
In the case where the elliptic fibration is a Calabi-Yau threefold, 
the curve $S$ supporting the gauge group always has genus $g=1+6K^2$ for the  SO($3$)-model and $1+K^2$ for the SO($5$) and SO($6$)-model. 
Since the number of hypermultiplets transforming in the adjoint representation is given by the genus of the curve $S$, this implies that these models cannot be defined with $S$ a rational curve and hence always has matter transforming in the adjoint representation.

\subsection{Summary of results}

In physics, a $G$-model provides  a geometric engineering of  gauge theories, related to F-theory and M-theory  compactifications, yielding a gauge group $G$ and matter fields transforming in the representation $\mathbf{R}$ of $G$. 
The chambers of the hyperplane arrangement I($\mathfrak{g}, \mathbf{R})$ correspond to the Coulomb branches of  a five-dimensional $\mathcal N=1$ (eight supercharges) gauge theory with gauge group $G$ and hypermultiplets transforming in the representation \textbf{R}. 
 The flats  of the hyperplane arrangement  are the mixed Coulomb-Higgs branches of the gauge theory. The faces of the hyperplane arrangement are identified with partial resolutions of the Weierstrass model corresponding to the $G$-model.

For each of the $G$-models that we consider in this paper, we do the following: 
\begin{enumerate}
\item We discuss their crepant resolutions and a classification of the singular  fibers of the resolved geometries. 
The SO($5$)-model is the only one that has non-Kodaira fibers, namely, a fiber of type 1-2-1, that is a contraction of a Kodaira fiber of type I$_0^*$.
There is a unique crepant resolution for the SO($3$)-model and the SO($5$)-model; and two crepant resolutions connected by a flop for the SO($6$)-model. 
 While it is known that  Spin($9+n$)-models that are crepant resolutions of Weierstrass models are usually not flat\footnote{Non-flatness here means that the fiber can have higher dimensional components.},  Spin($7$) and Spin($8$)-models can be given by flat fibrations that are crepant resolutions of Weierstrass models \cite{G2}. We show that the crepant resolution of the Weierstrass models defining   SO($3$), SO($5$),   and SO($6$)-models also give flat fibrations. 

\item The understanding of the crepant resolutions allows us to compute topological data such as the Euler characteristic of these varieties \cite{Euler} and other characteristic invariants \cite{Char1}, the triple intersection numbers of their fibral divisors, and in the Calabi-Yau threefold case, the Hodge numbers.

\item The discriminant locus has an  irreducible component $S$ over which the generic fiber is reducible while the fibers away from $S$ are irreducible. 
The generic fiber over $S$ is respectively  I$_2^{\text{ns}}$ or III for SO($3$),  I$_4^{\text{ns}}$ for SO($5$), and I$_4^{\text{s}}$ for SO($6$)-models.
 The discriminant of the SO($3$)-model contains two irreducible components, one of type  I$_2^{\text{ns}}$ and one of type I$_1$. 
The discriminant of the SO($5$) model (resp. the SO($6$)-model) contains three irreducible components meeting at the same locus, one of type  I$_4^{\text{ns}}$ (resp. I$_4^{\text{s}}$) and the two others of  type I$_1$; all three components meet at the same codimension-two locus.  

The singular fibers over the intersections of these components are responsible for the representation $\mathbf{R}$.
The representation is determined by computing intersection numbers of the components of the codimension two fibers with the fibral divisors. 
These intersection numbers are interpreted as minus the weights of a representation. The full representation is derived from few weights only by using the notion of {\em saturated set of weights}.  
The adjoint representation is always present and we determine the non-adjoint components of $\mathbf{R}$ geometrically. 
The SO($3$)-model does not have any representation other than the adjoint representation. In contrast, the SU$(2)$-model will have a fundamental representation $\mathbf{2}$. 
For  the SO($5$)-model, the collision I$_4^{\text{ns}}+\text{I}_1$ gives a non-Kodaira fiber of type 1-2-1, which is an incomplete fiber of type I$_0^{*\text{ss}}$ and corresponds to the vector representation $\mathbf{5}$.
$$
 \text{I}_4^{\text{ns}}+\text{I}_1 +\text{I}_1\to 1-2-1.
$$
For  SO($6$)-model, the collision I$_4^{\text{s}}+\text{I}_1$ gives a Kodaira fiber of type I$^{*\text{s}}_0$ whose dual graph is $\widetilde{D}_4$ and 
 corresponds to the vector representation $\mathbf{6}$. 
$$
 \text{I}_4^{\text{s}}+\text{I}_1\to \text{I}^{*\text{s}}_0.
$$
While the Weierstrass models of the SO($3$) or SO($5$)-model have unique crepant resolutions, the SO($6$)-model has two crepant resolutions connected by an Atiyah-like flop.

\item We identify the representation $\mathbf{R}$ from codimension two degenerations of the generic singular fiber using intersection theory. 
  We then  study the hyperplane arrangement I($\mathfrak{g}, \mathbf{R}$). 
  In the SO($3$)-model, the representation $\mathbf{R}$ is the adjoint, thus  I($\mathfrak{g}, \mathbf{R}$) has only one chamber that is the full dual fundamental Weyl chamber of A$_1$. 
  In the SO($5$)-model, the representation $\mathbf{R}$ is the direct sum of the adjoint and the vector representation of SO($5$), since all the weights of the vector representations are also weights of the adjoint representation,  I($\mathfrak{g}, \mathbf{R}$) has only one chamber that is again the full dual fundamental Weyl chamber of A$_1$.
  In the SO($6$)-model, the representation $\mathbf{R}$ is also the  direct sum of the adjoint and the vector representation of SO($6$). However, the vector representation of SO($6$) has two weights whose kernel intersects the interior of the dual fundamental Weyl chamber of A$_3$. It follows that we have two chambers separated by the wall orthogonal to the weight $(1,0,-1)$. 
  
 \item 
In the case where the elliptic fibration is a Calabi-Yau threefold, we study compactifications of M-theory; these compactifications can be used to geometrically engineer a 5D $\mathcal N=1$ (i.e. eight supercharges) supergravity theory with gauge group $G$ and hypermultiplets in the representation \textbf{R}. 
 This ${\cal  N}=1$ 5D supergravity theory has a Coulomb branch, and the dynamics of the theory on the Coulomb branch are controlled by a one-loop exact prepotential  $\mathcal{F}_{\text{IMS}}$. 
We  determine the number of hypermultiplets transforming in each of the irreducible components of $\mathbf{R}$ by comparing the exact prepotential 
of a 5D gauge theory with gauge group $G$ and an undetermined number of hypermultiplets transforming in the representation $\mathbf{R}$, with the triple intersection numbers of the fibral divisors of the elliptic fibration. 
 By making this comparison, we determine explicitly the multiplicity of each irreducible component $\mathbf{R}_i$ in the 5D hypermultiplet representation \textbf{R} in terms of the intersection ring of the elliptic fibration.
\item 
The elliptic fibrations studied here can also be used to define F-theory compactifications geometrically engineering 6D $(1,0)$ supergravity when the elliptic fibration is a Calabi-Yau threefold. 
Since 6D $(1,0)$ theories 
are chiral, the cancellation of anomalies is an important 
consistency condition. We determine the number of charged hypermultiplets in the 6D theory by expressing the anomaly cancellation conditions in terms of geometric quantities and verifying that the allowed matter content matches what we find in the 5D $\mathcal N=1$ theory. 
\end{enumerate}

\section{ The SO(3)-model}
\label{sec:SO3}

The SO($3$)-model is given by a crepant resolution of the generic Weierstrass model with Mordell--Weil torsion $\mathbb{Z}/2\mathbb{Z}$. 
Without loss of generality, we can put the generator of the Mordell--Weil group at the origin $x=y=0$. The Weierstrass model of an SO($3$)-model  can then always be defined by the equation
	\begin{align}
	\label{eqn:basicSO3}
	Y_0 ~:~ y^2z -x( x^2 + a_2 xz + a_4  z^2)=0.
	\end{align}
We assume that the varieties  $V(a_2)$ and $V(a_4)$ are smooth irreducible varieties  intersecting transversally. In particular, $a_4$ is not a perfect square.

	The Mordell--Weil group of the elliptic fibration $\varphi:Y_0\to B$ is isomorphic to $\mathbb{Z}/2\mathbb{Z}$ and generated by the section  $\Sigma_1: x=y=0$.  
	The neutral element of the Mordell--Weil group is the section $\Sigma_0:x=z=0$.  
	 Both elements of the Mordell--Weil group are on the   line $x=0$ tangent to the generic fiber at $x=y=0$.  

Using equations \eqref{Eq:Formulaire}, the short form of the Weierstrass equation is specified by
$$
c_4=16(a_2^2-3a_4), \quad c_6=-32 a_2 (2 a_2^2-9 a_4).
$$
 The discriminant of the elliptic fibration \eqref{eqn:basicSO3}  is then
	\begin{align}
	\label{eqn:SO3disc}
		\Delta =a_4^2 (4 a_4 - a_2^2).
	\end{align}
The following theorem describes the fiber geometry of a minimal crepant resolution of an SO(3)-model. 
\begin{thm}
The reduced discriminant of a SO(3)-model consists of two smooth divisors $V(4 a_4 -a_2^2)$ and $V(a_4)$. 
The intersection $V(a_2,a_4)$ of these two divisors is non-transverse and corresponds to the cuspidal locus  of the elliptic fibration. 
The singular Weierstrass model $Y_0$  has three types of singular fibers: 
\begin{enumerate}
\item The generic fiber over $V(a_4)$ is of  type I$_2^{\text{ns}}$ (two rational curves meeting transversally along a divisor of degree 2 such that a quadratic field extension is needed to split the divisor into two closed rational points).
\item  The generic fiber over  $V(4a_4 -a_2^2)$ is of type I$_1$ (a nodal cubic). 
\item  The generic fiber over the collision $V(a_2,a_4)$ is  of Kodaira type III (two rational curves intersection  along a double point):
\end{enumerate}

\end{thm}

\begin{center}
\begin{tabular}{|c|c|}
\hline 
Locus &  Fiber  \\ 
\hline 
$V(a_4)$ &  I$_2^{\text{ns}}$ \\
\hline 
$V(a_2^2-4a_4)$ & I$_1$ \\
\hline 
$V(a_2, a_4)$ & III  \\
\hline 
\end{tabular}
\end{center}

\begin{proof}

The cuspidal locus is by definition $V(c_4,c_6)$.  	A direct computation gives 
$b_2=4a_2$, $c_4=16 (a_2^2-3 a_4)$,  $c_6=-32 a_2 (2 a_2^2-9 a_4)$, and thus it follows that the cuspidal locus is $V(a_2, a_4)$. 
Since the valuations of $(c_4,c_6, \Delta)$ over $V(a_4)$ and   $V(4 a_4 -a_2^2)$ are respectively  $(0,0,2)$ and $(0,0,1)$, it follows from  Tate's algorithm that  the generic fiber over the divisor $V(a_4)$ is Kodaira  type I$_2$ and the generic fiber over $V(4 a_4 -a_2^2)$ is Kodaira type I$_1$ (a nodal elliptic curve). 
(A fiber of type I$_2$ is characterized by two rational curves intersecting transversally at two closed points---Kodaira fibers of type I$_n$ are described in  Step 2 of Tate's algorithm.) 
In Tate's algorithm one classifies geometric fibers, and in order
to describe the geometric fiber,  it is necessary to work in the splitting field of the polynomial  $s^2-a_1 T + a_2$ whose discriminant is $b_2=a_1^2 -4 a_2$. 
If $b_2$ is a perfect square on the divisor over which we have the I$_n$ fiber, the fiber is split; otherwise, the fiber is non-split. 
In the present case, $a_1=0$ and $b_2 = -4a_2$ is not a perfect square modulo $a_4$, hence  the generic fiber over $V(a_4)$ is a non-split I$_2$ fiber, i.e. type I$_2^{\text{ns}}$. 
We also expect to see the fiber degeneates at the intersection of the two components of the discriminant locus since that intersection is supported on the cuspidal locus of the fibration. 
We prove in the next Section \ref{Sec:FG}  that the fiber over $V(a_2,a_4)$ is Kodaira type III (a type III fiber is composed of two rational curves meeting at a double point.) 
This is a result of the fact that the type I$_2$ fiber, which is located over $V(a_4)$, degenerates to a type III fiber over $V(a_2)$.
\end{proof}

\begin{cor}
If  $a_2$ is a perfect square, then the fiber over $V(a_4)$ is type  I$_2^{\text{s}}$. An SO(3)-model with an I$_2^\text{s}$ fiber over $V(a_4)$ can always be described as:
\begin{equation}
\text{SO(3)-model with I$_2^\text{s}$}: \quad Y_0: zy^2 + a_1 x y z  -x^3-a_4 x z^2=0
\end{equation}
\end{cor}

\begin{cor}
If  $a_2$ is identically zero, then the fiber over $V(a_4)$ is type III. An \textnormal{SO($3$)}-model with a type III fiber can always be described as:
\begin{equation}
\text{SO(3)-model with III}: \quad Y_0: zy^2  -x^3-a_4 x z^2=0
\end{equation}
with $j$-invariant, $j=1728$ over $V(a_4)$. 
\end{cor}
\subsection{Singularities of the Weierstrass model}

A Weierstrass model can only be singular away from its zero section $x=z=0$.  For that reason, we discuss the singularities in the open patch $z\neq 0$. 
The Weierstrass equation \eqref{eqn:basicSO3} is then a double cover of the base $B$  branched along  $y=x(x^2+a_2 x + a_4)=0$. 
The branch locus consists of two irreducible divisors  $V(y,x)$ and $V(y, x^2+a_2 x+ a_4)$ meeting transversally in codimension two along  $V(y,x,a_4)$.
We notice that the divisor $V(y,x)$ is the generator of the Mordell--Weil group on each smooth curve. 
Since the singularities of a double cover are those of its branch locus, it is immediate to see that the singular  scheme 
 is supported on the locus 
	\begin{align}
		\text{Sing}(Y_0) : \quad  V(x , y, a_4 ).
	\end{align}

\subsection{Crepant resolution}
\label{SO3res}
The minimal crepant resolution of the SO(3)-model requires only a single blowup. 
We can take the center of this blowup to be either of the ideals $(x,y)$ or $(x,y,a_4)$. 
We denote by $X_0=\mathbb{P}[\mathscr{O}_B\oplus \mathscr{L}^{\otimes 2}\oplus\mathscr{L}^{\otimes 3}]$ the original ambient space, which is the $\mathbb{P}^2$ bundle. 
We obtain a crepant resolution by blowing up along the ideal  $(x,y)$.  We denote by $X_1$ the blowup of $X_0$ along the ideal $(x,y)$. 
	The blowup can be implemented explicitly by the  substitution $(x,y) = (x_1 e_1, y_1 e_1)$ where $V(e_1)=E_1$ is the exceptional divisor, which is a $\mathbb P^1$ bundle with fiber  parametrized by $[x_1:y_1]$.
	The projective coordinates describing $X_1$ are 
	\begin{equation}
[ x=e_1 x_1: y=e_1 y_1: z], [ x_1:y_1].
\end{equation}	
The fibers of $X_1\longrightarrow B$ are Hirzerbuch surfaces of degree one. 
\begin{thm}
The blowup 
\begin{align}
		\begin{tikzpicture}
				\node(X1) at (0,0) {$f:\quad X_0$};
				\node(X0) at (3.5,0) {$X_1$} ;
				\draw[big arrow] (X0) -- node[above,midway]{$(x,  y | e_1)$} (X1);
			\end{tikzpicture}
	\end{align}
	provides a crepant resolution of the  SO(3)-model given by the Weierstrass equation  \eqref{eqn:basicSO3}.

The proper transform of the Weierstrass model $Y_0$ under the blowup is   the elliptic fibration $\varphi_1: Y\to B$ given as the following hypersurface in $X_1$: 
	\begin{align}
	\label{eqn:SO3E1}
		Y ~:~ e_1 y^2_1z - x_1 ( e_1^2 x^2_1 +  e_1 a_2x_1z+ a_4z^2 )=0. 
	\end{align}

	\end{thm}

The section $\Sigma_0=V(x,z)=V(x_1,z)$. The generating section of the Mordell--Weil group is  $\Sigma_1=V(e_1,x_1)$.
\subsection{Fiber geometry}\label{Sec:FG}
\begin{thm}
The fiber I$_2^{\text{ns}}$ over $V(a_4)$ is composed of the following two curves: 
	\begin{subequations}
	\begin{align}
		C_0 :&\quad   a_4= y^2_1z -x^2_1(e_1 x_1 +a_2 z)=0, \quad  && [ e_1 x_1 : e_1 y_1 :z] [x_1: y_1].\\
		C_1 :&\quad  a_4=e_1=0,\quad && [0: 0:z][x_1:y_1].
	\end{align}
	\end{subequations}
	\end{thm}
	The curve $C_0$ is the proper transform of the original elliptic fiber and corresponds to the usual normalization of a nodal curve\footnote{	 For the generic curve $C_0$,  $z$ and $x_1$ are units.  To show that $C_0$ is a rational curve, we show that is has a rational parametrization  in the patch $z e_1 \ne 0$. 
	Then fixing $z=1$, introducing the variable $t = y_1/x_1$, and solving for $e_1 x_1$, we find a rational parametrization of $C_0$ given by $t \mapsto [e_1 x_1: e_1 y_1 : 1][1:y_1/x_1]=[t^2 -a_2:t(  {t}^2 -a_2) : 1][1:t]$. 
}.  $C_1$ is the exceptional curve coming from the blowup and is parametrized by $[x_1: y_1]$---notice that $C_1$ corresponds to the zero section of the Hirzebruch surface (i.e. $\mathbb{F}_1$) fiber of $X_1$. 
	Both $C_0$ and $C_1$ are smooth rational curves parametrized by $[x_1:y_1]$.    
When we restrict to $V(a_4)$, we see that the section $V(x_1,z)$ is on  $C_0$ while the section $V(x_1, e_1)$ is located on $C_1$. 
 The generator of the Mordell--Weil group touches the curve $C_1$ but not $C_0$ \cite{Mayrhofer:2014opa}. 
The intersection of the irreducible components $C_0$ and $C_1$ is a zero-cycle of degree two:
	\begin{align}
	\label{eqn:C0C1}
		C_0 \cap C_1 : \quad   a_4=e_1  =y^2 - a_2 x^2 =0\quad \quad [0:0:z][x:y],
	\end{align}
As we move over  $V(a_4)$, $C_0 \cap C_1$ defines a double  double cover of $V(a_4)$ branched at $V(a_4,a_2)$, which splits into two distinct divisors in the special case that $a_2$ is a perfect square.
  Over the branch locus $V(a_2,a_4)$, the intersection \eqref{eqn:C0C1} collapses to a double point (i.e. a tangent) and hence the generic type I$_2^{\text{ns}}$ fiber degenerates to a type III fiber---see Figure \ref{fig:SO3int}.

We now turn our attention to the divisors $D_0$ and $D_1$ swept out by $C_0$ and $C_1$ as they move over the base. 
Notice that $C_0$ is both smooth and birationally-equivalent to $\mathbb P^1$, while $C_1$ is clearly a $\mathbb P^1$. Hence, both $D_0$ and $D_1$ are  projective bundles of the type $\mathbb P[ O_S\oplus \mathscr{L}]$ where $S=V(a_4)$:
	\begin{align}
	\label{eqn:P1bundle}
		D_0 \cong \mathbb P_{S}[ \mathscr{O}_S \oplus  \mathscr L], \quad D_1 \cong \mathbb P_{S}[ \mathscr {O}_S \oplus  \mathscr L].
	\end{align}

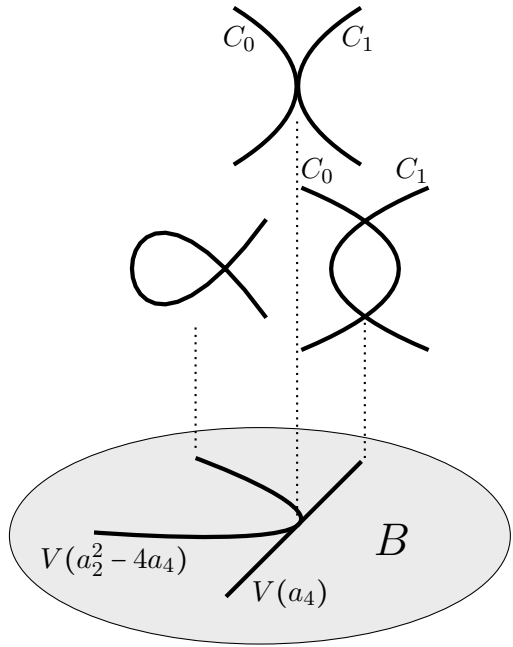
		\begin{figure}[htb]
\begin{center}
			$
			\begin{array}{c}
				\begin{array}{c}
				\begin{tikzpicture}[scale=.9]
					\draw[fill=black!8!] (.45,.85) ellipse (3.7cm and 1.6cm);
					\node at (2.4,.8) {\LARGE $B$};
					\node at  (.9,0) {$V(a_4)$};
					\node at  (-1.7,.5) {$V(a_2^2 - 4a_4)$};
					\node(a) at (-.2,-.2) {};
					\node(b) at (2.1,2.1) {};
					\draw[ultra thick] (a) to (b);
					\node(e) at (1,1) {};
					\node(c) at (-2,.9) {};
					\node(d) at (-.5,2){};
					\draw[dotted,thick] (e) to (1,7){};
					 \draw[ultra thick] plot [smooth,tension=1] coordinates {(d) (e) (c)};
					 \draw[dotted,thick] (d) to (-.5,4){};
					  \draw[dotted,thick] (2,2) to (2,4){};
					  \node at (2,4.8) {$ 
					  	\begin{tikzpicture}[scale=2]
							\draw[yshift=20,scale=.45,domain=-1.2:1.2,variable=\x, ultra thick] plot({\x*\x-.5,\x});
							\draw[yshift=20,scale=.45,domain=-1.2:1.2,variable=\x, ultra thick] plot({1-\x*\x-.5,\x}); 
						\end{tikzpicture}
								$};
					\node at (-.45,4.8) {$
							\begin{tikzpicture}[scale=3.1]
									\draw[ultra thick,scale=.4,domain=-1.2:1.2,variable=\x] plot({\x*\x-1,\x*\x*\x-\x-5});
							\end{tikzpicture}
								$};
						\node at (1,7.5) {$
						\begin{tikzpicture}[scale=2.9]
						\draw[ultra thick,xshift=6.5,yshift=20,scale=.45,domain=-.8:.8,variable=\x] plot({\x*\x-.5,\x});
							\draw[ultra thick,xshift=-6.5,yshift=20,scale=.45,domain=-.8:.8,variable=\x] plot({1-\x*\x-.5,\x}); 	
						\end{tikzpicture}
						$};	
						\node at (1.9,8.2) {$C_1$};
						\node at (.15,8.2) {$C_0$};
						\node at (2.7,6.25) {$C_1$};
						\node at (1.3,6.25) {$C_0$};
				\end{tikzpicture}
			\end{array}\\\\
			\begin{array}{|l|c|}
			\hline
				\text{Weierstrass model}  &y^2 z = x^3 + a_2 x^2 z + a_4 x z^2\\\hline
				\text{Discriminant} & \Delta=  a_4^2 (4 a_4 - a_2^2) \\\hline
				\text{Matter representation}  &\text{adjoint}   \\\hline
				\text{Representation multiplicity} & n_\text{\textbf{3}} = 1+ 6 K^2 \\\hline
				\text{Euler characteristic} & \frac{ 12 L}{1+ 4 L} c(TB)\\\hline
				\text{Triple intersections} &6 \mathcal F =~ 48 L^2( -  \alpha_0^3 + \alpha_0^2 \alpha_1 +  \alpha_0 \alpha_1^2 - \alpha_1^3)
				\\\hline
			\end{array}
			\end{array}
			$
			\end{center}
			\caption{Summary of geometry for the resolution of the SO($3$)-model given by (\ref{eqn:SO3E1}). Note that the K\"ahler cone only consists of a single chamber.}
			\label{fig:SO3int}
		\end{figure}

\subsection{Euler characteristic and intersection products}

In this section, we compute the generating function for the Euler characteristic along with the triple intersections $D_i D_j D_k$. 
The Euler characteristic is the degree of the zero component of the  (homological) total Chern class. The result is given by the following theorem. 
\begin{thm}[{See \cite{Euler}}]

	The generating function for the Euler characteristic of an SO($3$)-model is
	\begin{align}
\varphi_*c(Y) =\frac{ 12 L t}{1+ 4 L t} c_t(TB),
	\end{align}
	where $c_t(TB)=\sum c_i (TB) t^i$ is the  Chern polynomial of the base $B$.
	\end{thm}
The Euler characteristic of an SO(3)-model over a base $B$ of dimension $d$ is the coefficient $t^d$ for $d=\dim B$ of the generating function. See Table \ref{Table:EulerSO3} for some examples.
We note that the generating function is the same as that of an elliptic fibration of type  E$_7$ studied in \cite{AE2}.	

\begin{table}	
	\begin{center}
	
	\begin{tabular}{|c|c|c|}
\hline
	$\dim B$ & Euler characteristic  & Calabi-Yau case\\
	\hline 
	$1$ & $12 L$  & $12 c_1$ \\
	\hline
	$2$ & $12 (c_1 - 4 L) L$ & $-36 c_1^2$ \\
	\hline 
	$3$ & $12 L (c_2 - 4 c_1 L + 16 L^2)$ & $12 c_1 (12 c_1^2 + c_2)$\\
	\hline 
	$4$ & $12 L (c_3 - 4 c_2 L + 16 c_1 L^2 - 64 L^3) $ & $12 c_1 (c_3-48 c_1^3 +-4 c_1 c_2)$ \\
	\hline 
	\end{tabular}
	\end{center}
	\caption{Euler characteristic of the SO(3)-model for bases of dimension up to $4$. 
	The $i$th Chern class of the base is  denoted  $c_i$. 
	The Calabi-Yau cases are obtained by imposing $L=c_1$. }
	\label{Table:EulerSO3}
\end{table}

We now turn our attention to the triple intersections $D_a D_b D_c$. 
\begin{thm}
Assume that the base $B$ is a smooth surface. The triple intersection numbers for the fibral divisors  $D_0$ and $D_1$ for an SO(3)-model are  
\begin{subequations}
\begin{align}
6\mathcal F &=(\alpha_0 D_0 +  \alpha_1 D_1 )^3 =48 L^2( -  \alpha_0^3 + \alpha_0^2 \alpha_1 +  \alpha_0 \alpha_1^2 - \alpha_1^3). 
	\end{align}
	\end{subequations}
	\end{thm}
\begin{proof}
The first step in computing intersection products is to  identify the classes of the divisors $D_0$ and $D_1$. 
We exploit the linear relations relating $D_0$ and $D_1$ with the sections $\Sigma_0$ and $\Sigma_1$ coming from Mordell--Weil group:
 \begin{equation}
 \Sigma_0 = V(z,x_1), \quad \Sigma_1=V(e_1, x_1).
 \end{equation}
  $\Sigma_0$ is the divisor corresponding to the zero section $V(x_1,z)$ and $\Sigma_1=V(e_1, x_1)$ is the divisor corresponding to the 
 generator of the Mordell--Weil group. 
	Using the equation of the elliptic fibration we find that:
\begin{align}
	(z)=3\Sigma_0,\quad  (x_1) =\Sigma_0+\Sigma_1.
\end{align}
From these linear relations we get:
\begin{align}
	\Sigma_0 = \frac{H}{3},\quad  \Sigma_1= -E_1+\frac{2}{3}H+2L,
\end{align}
where we have used $(x_1)=2L+H-E_1$ and  $(z)=H$.  
 According to the defining equation (\ref{eqn:SO3E1}) for   {$Y$},  $D_0, D_1$ must satisfy the linear relations
\begin{align}
 E_1=\Sigma_1 + D_1,\quad D_0 + D_1 = 4 L.
\end{align}
That is 
\begin{align}
D_1 =2E_1 - \frac{2}{3} H - 2L, \quad \quad \quad D_0 = -2E_1  +\frac{2}{3}H+6L.
\end{align}
We  compute the triple intersection numbers by using two pushforwards \cite{Euler}: 
\begin{equation}
D_a D_b D_c = \pi_* f_* (D_a D_b D_c). 
\end{equation}

\end{proof}

\subsection{Counting 5D matter multiplets.}

Given a $G$-model with representation \textbf{R}, knowledge of the triple intersection numbers provides a derivation of the number of matter multiplets for a 5D  supergravity theory with gauge group $G$ and hypermultiplets in the representation \textbf{R} in terms of geometric data. 
Since we know from looking at the weights of the vertical curves of the elliptic fibration that the only possible matter representations are in the adjoint representation, we compute the Intriligator-Morrison-Seiberg (IMS) potential for 
 a gauge theory with Lie algebra A$_1$ and $n_\textbf{adj}$ hypermultiplets transforming in the adjoint representation. 
The one loop quantum correction to the prepotential is\footnote{We will take $\alpha$ in the basis of fundamental weights and $\phi$ in the basis of simple coroots.} (with mass parameters set to zero): 
\begin{equation}
6\mathcal{F}_{\text{IMS}} =\frac{1}{2}\Big(\sum_ {\alpha} |(\alpha, \phi)|^3 -n_\textbf{adj}\sum_{\alpha  } |(\alpha,\phi)|^3\Big)=8(1-n_\textbf{adj})\phi^3,
\end{equation}
where $\phi>0$ is in the the dual of the fundamental Weyl chamber and $\alpha$ are the weights of the adjoint representation of $A_1$. 
This is supposed to match the computation of the triple intersection number $D_1^3=-48 K^2$. 
By a direct comparison, we get: 
\begin{equation}
n_\textbf{adj}=1+6  K^2. 
\end{equation}
In particular, the number $n_{\textbf{adj}}$ is never zero as $K^2$ is an integral number for any smooth surface $B$. 
Assuming the Calabi-Yau condition, the curve $V(a_4)$ has Euler characteristic $\chi$ and genus $g$ with 
\begin{equation}
\chi=(-K+4K) (-4K)=-12 K^2, \quad 
g=1-\frac{1}{2}\chi= 1+ 6 K^2.
\end{equation}
We can then express the number of adjoints in terms of the genus of the curve $V(a_4)$:
\begin{equation}
n_\textbf{adj}=g.
\end{equation}

\begin{rem}
We can derive the above result geometrically as well. 
In a Calabi-Yau, for a $\mathbb{P}^1$ fibration over a curve of genus $g$, $D^3=K_D^2\cap [D]=8(1-g)$. 
For the second equality, see \cite{Hartshorne}[{Chap. V,  Corollary 2.11}]. 
Assuming that this is equal to the coefficient of $6\mathcal{F}$, we have $n_\textbf{adj}=g$.  
\end{rem} 

\section{The SO(5)-model} 
\label{sec:SO5}
\begin{defn}
	The SO(5)-model is a Weierstrass model with defining equation
	\begin{align}\label{SO5eqn}
		Y_0 ~:~ y^2z  =( x^3 + a_2 x^2z + s^2 xz^2).
	\end{align}
	\end{defn}
	\noindent{\bf Assumptions} : The divisors $V(a_2)$ and  $V(s)$ of $B$ are generic. In particular, they are smooth, not  proportional to each other nor is $a_2$ a perfect square modulo $s$. 
	
	\begin{rem}
The SO(5)-model is  is obtained from the SO(3)-model by way of the substitution 
$a_4 \longrightarrow s^2$ where $s$ is a section of $\mathscr{L}^{\otimes 2}$. 
This specialization does not  alter the Mordell--Weil group, which is still $\mathbb Z/2\mathbb Z$ and generated by the point $x=y=0$ of the generic fiber. 
\end{rem}

\begin{lem}
The discriminant and the cuspidal locus of  the elliptic fibration \eqref{SO5eqn} are respectively 
	\begin{align}
	\label{eqn:SO5disc}
		\Delta &  
		=  s^4 (2s + a_2)(2s-a_2), \quad 
V(c_4,c_6)= V(s,a_2).
	\end{align}
	\end{lem}
		We denote the divisor $V(s)$ of $B$ as $S$.  
	The three components of the reduced discriminant intersect pairwise transversely along the cuspidal locus $V(  {s},a_2)$. 
Tate's algorithm indicates that over the generic point of $S=V(s)$ the generic geometric fiber is type I$_{4}$, while the geometric fiber over the generic points of the two other  components of the discriminant  is a Kodaira fiber of type I$_1$ (i.e. a nodal curve). Since by assumption $a_2$ is not a perfect square modulo $s$, the generic fiber over $S$ is actually not geometrically irreducible; rather, the fiber type is I$_4^{\text{ns}}$. Such a fiber has as its dual graph the extended Dynkin diagram $\widetilde{\text{C}}_2^t$.

\subsection{Crepant resolution}

The singular locus of the elliptic fibration defined in \eqref{SO5eqn} is supported on  
	\begin{align}
		\text{Sing}(Y_0) &= V(x,y ,s). 
	\end{align}

\begin{thm}\label{thm:res}
	A crepant resolution  $f:Y\longrightarrow Y_0$ of the elliptic fibration \eqref{SO5eqn}  is given by the following sequence of  blowups along smooth centers:
	
	\begin{align}
		\begin{array}{c}
			\begin{tikzpicture}
				\node(E0) at (0,0) {$X_0=\mathbb{P}[\mathscr{O}_B\oplus \mathscr{L}^{\otimes 2}\oplus \mathscr{L}^{\otimes 3}]$};
				\node(E1) at (6,0) {$X_1$};
				\node(E2) at (10,0) {$X_2$};
				\draw[big arrow] (E1) -- node[above,midway]{ $ (x,y,s|e_1)$} (E0);
				\draw[big arrow] (E2) -- node[above,midway]{ $(x,y,e_1|e_2)$} (E1);
			\end{tikzpicture}
		\end{array}
		\end{align}
		The relative projective coordinates of $X_0$, $X_1$, and $X_2$ are parametrized as follows:
\begin{equation}
[e_2^2 e_1 x: e_2^2 e_1 y :z] ,[ e_2  x:e_2 y:s] ,[x:y:e_1].
\end{equation}

		The proper transform of the elliptic fibration \eqref{SO5eqn}  is 
		\begin{align}
		\label{resSO5}
			Y : \quad y^2z - e_2^2 e_1 x^3 -a_2 x^2z - e_1 s^2 xz^2=0.
		\end{align}

\end{thm}

\begin{proof}

 We first blow up $X_0=\mathbb{P}[\mathscr{O}_B\oplus \mathscr{L}^{\otimes 2}\oplus \mathscr{L}^{\otimes 3}]$ along the ideal  $(x,y,s)$ and we denote the blowup space $X_1$.
   The exceptional divisor $E_1=V(e_1)$ is a  $\mathbb{P}^2$ projective bundle over $V(x,y,s)$. We abuse notation and implement the blowup by way of the substitution $(x,y,s)\mapsto (e_1 x_1,  e_1 y, e_1 s)$. 
When these substitutions are performed on the  defining equation of $Y_0$, we can factor out two powers of $e_1$ which ensures that the blowup is crepant for the elliptic fibration. 
  For the second  blowup $X_2=\mathrm{Bl}_{  {(}x,y,e_1  {)}} X_1$. The exceptional divisor  $E_2=V(e_2)$ is a $\mathbb{P}^2$ projective bundle over $V( x_1, y_1,e_1)$. 
   Once again we can factor out two powers of $e_2$ to obtain a crepant blowup. By studying the Jacobian criterion, we see that there are no singularities left if $V(a_2)$ is a smooth divisor. 
\end{proof}

\subsection{Euler characteristic and triple intersections}
	In this section, we compute the pushforward of the generating function for the Euler characteristic and triple intersection form of the SO(5)-model to $B$. 
The SO($5$)-model refers to the elliptic fibration $Y\to B$ defined by the crepant resolution given in Theorem \ref{thm:res}.

\begin{thm}
Assume that the base $B$ is a smooth surface. The triple intersection polynomial for the fibral divisors  $D_0$ and $D_1$ for a SO($5$)-model is
\begin{subequations}
\begin{align}
6\mathcal F &=(\alpha_0 D_0 +  \alpha_1 D_1+\alpha_2 D_2 )^3
=8 L^2 (-\alpha_0^3 + 3 \alpha_0 \alpha_1^2 - 4 \alpha_1^3 + 3\alpha_1^2 \alpha_2 -  \alpha_2^3).
	\end{align}
	\end{subequations}
	\end{thm}
\begin{proof}
Pushforward  $D_a D_b D_c [  {Y}] \cap [X_2]$ with  
	\begin{align}
		D_0 = 2L - E_1 ,~~  D_1 =E_1 - E_2 ,~~ D_2 = E_2,~~[Y]=3H+6L-2E_1-2E_2.
	\end{align}

\end{proof}

\begin{thm}[{See \cite{Euler}}]

	The generating function for the Euler characteristic of a SO($5$)-model is
	\begin{align}
\varphi_*c(Y) =\frac{ 4 L t (3+4Lt)}{(1+ 2 L t)^2} c_t(TB),
	\end{align}
	where $c_t(TB)=\sum c_i (TB) t^i$ is the  Chern polynomial of the base $B$.
	\end{thm}
	
\begin{table}	
	\begin{center}
	
	\begin{tabular}{|c|c|c|}
\hline
	$\dim B$ & Euler characteristic  & Calabi-Yau case\\
	\hline 
	$1$ & $12 L$  & $12 c_1$ \\
	\hline
	$2$ & $4 (3c_1 - 8 L) L$ & $-20 c_1^2$ \\
	\hline 
	$3$ & $4 ( 20 L^2 - 8 L c_1 + 3 c_2 ) L$ & $12 c_1(4c_1^2 + c_2)$\\
	\hline 
	$4$ & $4 (-48L^3 + 20 L^2 c_1 - 8 L c_2 + 3 c_3) L $ & $4 c_1 (3 c_3 - 28 c_1^3 - 8 c_1 c_2)$ \\
	\hline 
	\end{tabular}
	\end{center}
	\caption{Euler characteristic of SO(5)-model for bases of dimension up to $4$. 
	The $i$th Chern class of the base is  denoted  $c_i$. 
	The Calabi-Yau cases are obtained by imposing $L=c_1$. }
	\label{Table:EulerSO5}
\end{table}

\subsection{Fiber geometry}

\begin{cor} The elliptic fibration defined by the crepant resolution  described in Theorem \ref{thm:res}  is a flat elliptic fibration $\varphi_2: Y\longrightarrow B$. The fibral  divisors are:
		\begin{align}
			D_0 &: s =y^2 z- x^2(e_2^2 e_1 x + a_2z)=0 ,\quad\quad && [e_2^2 e_1 x : e_2^2 e_1 y: z],[e_2 x : e_2 y : 0] ,[ x:y: e_1]\\
			D_1&: e_1 = y^2 -a_2 x^2=0 ,\quad \quad&& [ 0:0:z],[e_2x :e_2 y : s] ,[x:y:0]\\
			D_2&: e_2 = y^2 -x(a_2 x + e_1 s^2z )=0,\quad\quad && [ 0:0:z],[0:0:s],[x :y:e_1]
		\end{align}
	\end{cor}	
\begin{rem}		
The total transform of the divisor $S=V(s)$ in $Y$ is $V(te_1 e_2)$ composed of three irreducible and reduced components. 
We denote the irreducible divisors  $D_0=V(s)$, $D_1=V(e_1)$, and $D_2=V(e_2)$ in $Y$. We denote their generic fibers over $S$ as (resp.) $C_0$, $C_1$, and $C_2$. 
The fibers  $C_0$ and $C_2$ are geometrically irreducible while $C_1$ is not. After a quadratic field extension $C_1$ splits into rational curves, namely $C_1^+$ and $C_1^-$ with fibers $C_1^+$ and $C_1^-$. 
The curves $C_0$, $C_1^+$, $C_2$, $C_1^-$ define a Kodaira fiber of type I$_4$ with dual graph the (untwisted) affine Dynkin diagram $\widetilde{A}_3$. 
The curves $C_0$, $C_1$, and $C_2$ have as their dual graphs the twisted affine Dynkin diagram $\widetilde{C}_2^t$ (in the notation of Carter), also denoted $\widetilde{D}_4^{(2)}$ in the notation of Kac---see Figure \ref{fig:SO5}.
\end{rem}

\noindent Over the degeneration locus $V(a_2)$, we find
\begin{align}
	Y~:~
\begin{cases}
C_0 \longrightarrow C_0 \\
C_1 \longrightarrow 2C_{1}' \\
C_2\longrightarrow C_2 
\end{cases}
\end{align}
where
\begin{align}
Y~:~
\begin{cases}
\begin{array}{llll}
		C_0 &: s= y^2 z- e_2^2 e_1 x^3 =0,~~ &&[e_2^2 e_1 x : e_2^2 e_1 y: z],[e_2 x : e_2 y : 0] ,[ x:y: e_1]\\
		2C_1'&: e_1 = y^2=0 ,\quad \quad&& [ 0:0:z],[e_2x :0 : s] ,[x:0:0]\\
			C_2&: e_2 = y^2 - e_1 s^2xz =0,\quad\quad && [ 0:0:z],[0:0:s],[x :y:e_1].
		\end{array}
		\end{cases}
	\end{align} 
Note that the conic $C_2$ remains non-degenerate in over the locus $V(a_2)$.

		\begin{thm}
The fiber over the generic point of the divisor $S=V(s)$ is of type I$_4^{\text{ns}}$ and degenerates along $V(s,a_2)$  to a non-Kodaira fiber of type $1-2-1$. 
The divisors $D_0$ and $D_2$ are isomorphic to the projective bundle  $\mathbb P_S[\mathscr O_S\oplus  \mathscr L]$ while 
$D_1$ is a double cover of $S\times \mathbb{P}^1$ and a branch locus over the cuspidal locus $V(s,a_2)$. 
Defining $S'$ as the double cover of $S$ branched over $V(s,a_2)$, $D_1$ is isomorphic to $\mathbb P_{S'}[\mathscr O_S\oplus  \mathscr L]$
:		
	\begin{align}
		D_0 \cong \mathbb P_S[\mathscr O_S\oplus  \mathscr L],\qquad 
		 D_1\cong \mathbb{P}_{S'}[ \mathscr O_{S'} \oplus   \mathscr L], 
\qquad		D_2\cong \mathbb P_{S}[ \mathscr O_S \oplus   \mathscr L ].
		\end{align}		
		\end{thm}
		\begin{proof}
		Since the center of the second blowup is away from $D_0$,  it is enough to stop at  the first blowup to understand the geometry of $C_0$ and $D_0$.  
By solving for $e_1$, we see that $C_0$ is the normalization of  a nodal curve. 
In the patch $x\neq 0$, an affine parameter of the rational curve $C_0$ is $y/x$.
It follows that the divisor $D_0$ with generic fiber $C_0$ is isomorphic to $\mathbb{P}_S [\mathscr{O}_S\oplus \mathscr{L}]$.

	The  curve $C_2$ is a conic defined in $\mathbb P^2$ with coordinates $[x :y :e_1]$. 
	The discriminant of the conic is $s^4/4$ which is a unit. Hence we again get a $\mathbb{P}^1$-bundle over $S$. 
	Working in the patch $x\neq 0$, after solving for $e_1$, we can parametrize the conic by  $y/x$. It follows that the divisor $D_2$ with fiber $C_2$ is isomorphic to $\mathbb{P}_S [\mathscr{O}_S\oplus \mathscr{L}]$.

To understand the geometry of $D_1$, we consider the proper morphism  $\pi: D_1\to S$. Since the fibers of $\pi$ are not connected, we consider the Stein factorization $\pi=f\circ \rho$  with $f: D_1\to S'$ a morphism with connected fibers and  $\rho:S'\to S$ the double cover of $S$ branched over $V(s,a_2)$. The Euler characteristic of $S'$ is $\chi(S')=2 \chi(S)- [S][a_2]=2(1-g)-4L^2=-4L^2-4L^2=-8L^2$ and we also retrieve $D^3_2=8(1-g')=4\chi(S')=-32L^2$.

 we go to a field extension where we can take the square root of $a_2$. It is then clear that $C_1$ is a double cover (branched along $V(a_2)$) of  a $\mathbb{P}^1$ bundle over $S$ with projective fiber paramatrized by 
$[e_2x:s]$. In particular, since $e_2x$ and $s$ are both sections of $\mathscr{L}^{\otimes 2}$ over $S$, the $\mathbb{P}^1$ bundle is trivial. Thus $D_2$ is the double cover of $S\times \mathbb{P}^1$ branched  along the  divisor $V(a_2)$ of $S$, that is along the cuspidal locus of the elliptic fibration.

\end{proof}

The geometric fiber of $D_2$ is composed of two non-intersecting projective bundle that coincide into a double line over the cuspidal locus $V(s,a_2)$ in  $B$.

The geometric fiber over the generic point of $S$ is a Kodaira fiber of type I$_4$. 
The description of the geometric fiber I$_4$ requires at least a quadratic field extension for the square root of $a_2$ to be well defined. The fiber over the generic point of $S$ has an affine Dynkin diagram of type $\widetilde{\text{C}}_2$, which corresponds to a type I$_4^{\text{ns}}$ fiber. 
The fiber degenerates over the cuspidal locus $V(a_2)\cap S$ where the curve $C_1$ degenerates into two coinciding lines giving a fiber of type $1-2-1$ which we can think of as an incomplete I$_0^*$.  
	\begin{thm}
	The intersection numbers between the divisors $D_a$ and their generic fibers $C_a$ give the following intersection matrix  
	\begin{equation}
\begin{pmatrix} D_a C_b \end{pmatrix} =
\begin{tabular}{r}
$\begin{matrix}
 C_0 & C_1  & C_2
\end{matrix}\   \  $\\
$\begin{matrix}
D_0 \\
D_1\\
D_2
\end{matrix}
\begin{bmatrix}
-2 & \   \  2 &\  \  0 \\
\   \  2 & -4 &\  \   2\\
\  \   0&\   \    2& -2
\end{bmatrix}$
\end{tabular}	
	\end{equation}
	
	More generally, we have 
	\begin{equation}
\varphi_{2*} \big(D_a \, D_b\,  \varphi^*_2 M\cap[{Y}]\Big)=
\begin{tabular}{r}
$\begin{matrix}
 D_0 & D_1  & D_2
\end{matrix}\   \  $\\
$\begin{matrix}
D_0 \\
D_1\\
D_2
\end{matrix}
\begin{bmatrix}
-2 & \   \  2 &\  \  0 \\
\   \  2 & -4 &\  \   2\\
\  \   0&\   \    2& -2
\end{bmatrix} 
$
\end{tabular} (S   M) \cap [B], \quad \text{where} \quad M \in A_*(B).
	\end{equation}
	\end{thm}
	
	\begin{proof}
	The first equation is obtained from the second one by choosing $M$ to be the generic point of the divisor $S$ of $B$.  The second equation is a direct pushforward computation with  
	$$D_0=2L-E_1,\quad D_1=E_1-E_2, \quad D_2=E_2,\quad [Y] =(3H+6L-2E_1-2E_2)$$ 
	We do the pushforward in three steps since $\varphi_2=\pi \circ f_1 \circ f_2$.  We first compute the pushforward $f_{2*}$ to the Chow ring of  $X_1$, then we compute the pushforward $f_{1*}$ to the Chow ring of $X_0$, and finally we compute the pushforward $\pi_*$ to the Chow ring of $B$. 
	\end{proof}
	
	\begin{figure}[htb]
	\begin{center}
			$
			\begin{array}{c}
			\begin{array}{c}
				\begin{tikzpicture}[scale=1]
					\draw[fill=black!8!] (0,0) ellipse (3.8cm and 1.7cm);
					\node at (-2.4,-.45) {$V(2s+a_2)$};
					\node at (2.4,-.45) {$V(2s-a_2)$};
					\node at (0,-1) {\LARGE $B$};
					\node at (3,.3) {$V(s)$};
					\draw[ultra thick] (-2.3,1) --++ (4.6,-2);
					\draw[ultra thick] (2.3,1) --++ (-4.6,-2);
					\draw[ultra thick] (-3,0) --++ (6,0);
					\draw[dotted,thick] (-1.3,.55) -- (-1.3,2.8);
					\draw[dotted,thick] (2.3,1) -- (2.3,3);
					\draw[dotted,thick] (-3,0) -- (-3,3.75);
					\draw[dotted,thick] (0,0) -- (0,3.75);
					\node at (-.3,5) 
						{
							\begin{tikzpicture}
								\node[draw,ultra thick,circle,xshift=200,label=left:{$C_0$}](C0') at (0,0) {};
								\node[draw,ultra thick,circle,xshift=200,scale=.65,label=left:{$C_1^{'}$}](C13) at (0,-1) {2};
								\node[draw,ultra thick,circle,xshift=200,label=left:{$C_2$}](C2) at (0,-2) {};
								\draw[ultra thick] (C0') -- (C13) -- (C2);
							\end{tikzpicture}
						};
					\node at (-3,5) 
						{
							\begin{tikzpicture}[]
								\draw[dashed,thick,red] (-1.25,-1+.25) to (1.25,-1+.25);
								\draw[dashed,thick,red] (-1.25,-1-.25) to (1.25,-1-.25);
								\draw[dashed,thick,red]  (-1.25,-1+.25)  to (-1.25,-1-.25);
								\draw[dashed,thick,red]  (1.25,-1+.25)  to (1.25,-1-.25);
								\node[draw,ultra thick,circle,label=left:{$C_0$}](C0) at (0,0) {};
								\node[draw,ultra thick,circle,label=left:{$C_1^+$}](C1+) at (-1,-1) {};
								\node[draw,ultra thick,circle,label=right:{$C_1^-$}](C1-) at (1,-1) {};
								\node[draw,ultra thick,circle,label=left:{$C_2$}](C2) at (0,-2) {};;
								\draw[ultra thick] (C0) -- (C1+) -- (C2) -- (C1-) -- (C0);
							\end{tikzpicture}
						};
					\node at (-1.3,3.1) 
						{
							\begin{tikzpicture}[scale=2]
								\draw[scale=.5,domain=-1.2:1.2,variable=\x, ultra thick] plot({\x*\x-1,\x*\x*\x-\x-5});
							\end{tikzpicture}
						};
					\node at (2.3,3.5) 
						{
							\begin{tikzpicture}[scale=2]
								\draw[scale=.5,domain=-1.2:1.2,variable=\x, ultra thick] plot({\x*\x-1,\x*\x*\x-\x-5});
							\end{tikzpicture}
						};
	
				\end{tikzpicture} 
			\end{array}\\\\
				\begin{array}{|l|c|}
				 \hline
				\text{Weierstrass model} & y^2 z = x^3 + a_2 x^2 z + s^2 x z^2\\\hline
				\text{Discriminant} & 	\Delta =s^4 (4 s^2 - a_2^2)\\\hline
				\text{Matter representations} &\text{adjoint + vector w/ geometric weight } (-2,1)\\\hline
				\text{Representation multiplicities} & n_{\textbf{10}} = 1 + L^2; ~~ n_\textbf{5} = 3 L^2 \\\hline
				\text{Euler characteristic} &  \frac{ 4 L ( 3 + 4 L ) }{(1+2L)^2} c(TB)\\\hline
				\text{Triple intersection form} &6 \mathcal F =~ -8 L^2 (\alpha_0^3+4 \alpha_1^3  +  \alpha_2^3)+ 24 L^2 \alpha_1^2(  \alpha_0 + \alpha_2)
				\\\hline
			\end{array}
			\end{array}
			$
			\end{center}
	\caption{Summary of the geometry of the resolution (\ref{resSO5}) of the SO($5$)-model.  Note that 
	all the crepant resolutions of Weierstrass model are isomorphic to each other. Thus, its extended K\"ahler cone consists of a single chamber. 
	}
	\label{fig:SO5}
	\end{figure}

	When the divisors $D_i$ are not all geometrically irreducible, we notice that the quadratic intersection numbers correspond to a twisted affine Dynkin diagram. 
In the present case, we get the twisted diagram $\widetilde{C}_2^t$ in the notation of Carter or $\widetilde{D}_4^{(3)}$ in the notation of Kac---see Figure \ref{fig:SO5}.
It is possible to read off the Cartan matrix for the ordinary Dynkin diagram from the incidence matrix by 
 rescaling the  $a$th row by $\frac{2}{D_a\cdot  C_a}$, and then 
deleting the affine node corresponding to $C_0$ and examining the components of the remaining $2\times 2$ block. 
 We obtain
	\begin{align}
	\begin{pmatrix} D_a \cdot C_b\end{pmatrix} = 
		\begin{tabular}{r} 
		$ \begin{matrix} C_1 &C_2 \end{matrix}~~~$  \\$
		\begin{matrix} D_1 \\ D_2 \end{matrix} \begin{bmatrix}-2 & 1  \\ 2 & -2\end{bmatrix}.
	$\end{tabular}
	\end{align}
The above matrix is minus the C$_2$ Cartan matrix.

\subsection{Matter representation}
\begin{table}[htb]
\begin{center}
\begin{tabular}{|c|c|c|}
\hline 
$\mathbf{5}$  & $\mathbf{10}$ & $\mathbf{4}$\\ 
\hline 
$
\begin{array}{c}
\boxed{\  0,\ \  1}\\
\\
\boxed{\  2,\ -1}\\
\\
 \boxed{\  0,\ \  0}\\
\\
\boxed{-2,\  \  1}\\
\\
\boxed{\  0,\ -1}\\
\end{array}
$
& 
$
\begin{array}{c}\\
\boxed{\  2,\ \  0}\\
\\
\boxed{\  0,\ \  1}\\
\\
\boxed{-2,\ 2}\quad  \boxed{\    2,\ -1}\\
\\
\boxed{\  0,\ \  0}\quad   \boxed{\  0,\ \  0}\\
\\
\boxed{\  2,\ -2}\quad  \boxed{-2,\ 1}\\
\\
\boxed{\  0,\ -1}\\
\\
\boxed{-2,\  \  0}\\ 
\\
\end{array}
$
& 
$
\begin{array}{c}
\boxed{\  1,\ \  0}\\
\\
\boxed{- 1,\   \  1}\\
\\
 \boxed{1,\  -1}\\
\\
\boxed{-1,\  \ 0}
\end{array}
$
\\
\hline 
\end{tabular}
\end{center}
\caption{Weights of the representations $\mathbf{5}$, $\mathbf{10}$, and $\mathbf{4}$ of C$_2$ expanded in the basis of fundamental weights. \label{Table:RepC2} }  
\end{table}

We now consider the matrix of weight vectors we obtained from the curves over the cuspidal locus $V(a_2,s)$. 
 We 
delete the first row and column (corresponding to the affine node):
	\begin{align}
	\begin{pmatrix} w_a(C_b) \end{pmatrix} = \begin{bmatrix} \varpi_1 & \varpi_2 \end{bmatrix} = \begin{tabular}{r} $\begin{matrix}C_1^\prime & C_2 \end{matrix}~~ $ \\ $ \begin{matrix} D_1 \\ D_2 \end{matrix} \begin{bmatrix} -2 & 2 \\ 1 & -2 \end{bmatrix} $ \end{tabular}.
	\end{align}
While the second column vector $ \varpi_2$ is a root of $\text{C}_2$, the first column $\varpi_1 = \boxed{-2,\  1}$ is a weight  of the  representation  $\mathbf{5}$ of $\text{C}_2$. 
The saturation of the singleton $\{ \varpi_1\}$ is precisely the set of weights of the representation $\mathbf{5}$ of C$_2$, with highest weight $\boxed{\ 0,\   1}$.  
The $\mathbf{5}$ of C$_2$ is defined as the non-trivial irreducible component  ${\bigwedge}^2_0$ in the decomposition of the  antisymmetric representation   ${~\bigwedge}^2 V_{{\text{C}_2}}={~\bigwedge}^2_0\oplus \mathbb{C}$.
Equivalently, the representation \textbf{5} is  the vector representation of $\mathfrak{so}(5)$. This representation is quasi-minuscule and self-dual. 
 The adjoint representation is $\mathbf{10}$ with highest weight $\boxed{2,\ 0}$.

\subsection{Counting 5D matter multiplets}
The variable $\phi\in \mathfrak{h}$ (where $\mathfrak h \subset \mathfrak g$ is the Cartan subalgebra) is expressed in the basis of simple coroots. 
The weights $\varpi$ and the roots  $\alpha$ are expressed in the canonically dual basis (i.e. the basis of fundamental weights).  

The relevant part of the prepotential for a 5D gauge theory with gauge algebra C$_2$ and hypermultiplets in the representations $\textbf{5}$ and $\textbf{10}$ is (with mass parameters set equal to zero) \cite{IMS}:
\begin{equation}
\mathcal{F}_{\text{IMS}}= \frac{1}{12}\left({\sum_{\alpha} \Big|{(\phi, \alpha)}\Big|^3 - \sum_{\mathbf{R}} n_{\mathbf{R}} \sum_{\varpi}  \Big|(\phi, \varpi)\Big|^3}\right).
\end{equation}

The open fundamental Weyl chamber is the subset of  $\mathfrak{h}$  with positive intersection with all the simple roots. The simple roots of $\mathfrak{so}_5$ are:
\begin{equation}
\boxed{2, -1 } \quad \boxed{-2,   2}.
\end{equation}
It follows that the (dual) open fundamental Weyl chamber is the cone of $\mathfrak{h}$ defined by:
\begin{equation}
2\phi_2 > 2 \phi_1 > \phi_2 > 0. 
\end{equation}
\begin{rem}
A weight defines a hyperplane through the origin that intersect the open fundamental Weyl chamber if and only if  both $\varpi$ and $-\varpi$ are not dominant weights.
\end{rem} 
For the SO(5)-model, the representations that we consider are the adjoint and the vector representation. The nonzero weights of the representation $\mathbf{5}$ are all dominant up to an overall sign. Hence, 
 I$({\text{C}_2}, \mathbf{10}\oplus \mathbf{5})$ has only a unique chamber \cite{EJJN2}. 

\begin{prop}
The prepotential for the Lie algebra 
$\mathfrak{so}_5$ with $n_\textbf{adj}$ matter multiplets transforming in the adjoint representation and $n_{\textbf{5}}$ matter multiplets transforming in the vector representation consists of a single phase. 
Explicitly, the prepotential is:

\begin{align*}
\begin{split}
		6 \mathcal{F}_{\text{IMS}} &= \left(- 8 n_{\mathbf{10}} - 8 n_{\mathbf{5} }+ 8 \right) \phi_1^3 + \left( 12 n_{\mathbf{5}} - 12 n_{\mathbf{10}} + 12 \right) \phi^2_1 \phi_2 + \left( - 6 n_{\mathbf{5}} + 18 n_{\mathbf{10}} - 18 \right) \phi_1 \phi_2^2   + \left( 8 - 8 n_{\mathbf{10}} \right) \phi_2^3 . 
	\end{split}
\end{align*}
\end{prop}

\begin{prop}
Under the identification $\phi_i =\alpha_i$, the prepotential $\mathcal{F}^\pm_{\mathrm{IMS}}$ matches $\mathcal{F}^\pm|_{{\alpha}_0=0}$ if and only if 
\begin{equation}
n_{\mathbf{5}}=3n_{\mathbf{10}}-3=3L^2, \quad n_{\mathbf{10}}=1+L^2. 
\end{equation}
If we impose the Calabi-Yau condition and assume the base is a surface, $1+K^2$ is the genus of the curve $V(s)$ of class $2L$ in the base and we have: 
\begin{equation}
n_{\mathbf{5}}=3g-3=3K^2, \quad n_{\mathbf{10}}=g=1+K^2. 
\end{equation}
In particular, the genus cannot be zero. 
\end{prop}
\begin{proof}
We can compare the potential $\mathcal{F}_{\mathrm{IMS}}$ with the triple intersection form $\mathcal{F}$  after setting ${\alpha}_0=0$.  
To have a match of the types of monomials present in the potential, we have to eliminate the coefficient of the term ${\alpha}_1{\alpha}_2^2$. 
This condition gives $n_{\mathbf{5}}=3(n_{\mathbf{10}}-1)$. We then get a  perfect match 
$\mathcal{F}_{\mathrm{IMS}}=\mathcal{F}$  by imposing $n_{\mathbf{10}}=1+L^2$. 
\end{proof}

\section{The SO(6)-model}
\label{sec:SO6}

\begin{defn}
The SO(6)-model is specified by the Weierstrass equation
	\begin{align}
	Y_0 ~:~ y^2z + a_1  x yz &= x^3 + ms x^2z + s^2 xz^2 ,
	\end{align}
where $m$ is a constant number different from $-2$, $0$, and $2$. The coefficient   $a_1$ is a generic section of $\mathscr{L}$, and $s$ is a smooth section of $\mathscr{L}^{\otimes 2}$. 
\end{defn}

\begin{rem}
\label{rem:SO6}
The SO(6)-model is obtained from the SO(5)-model by making the substitution $a_2 = \frac{1}{2}a^2_1 + m s$. This allows for the generic fiber to be of type I$_4^\text{s}$ and the discriminant to not have other reducible fiber types in codimension one. 
The Mordell--Weil group is still $\mathbb{Z}/2\mathbb{Z}$. 
\end{rem}

 \noindent The discriminant is
	\begin{align}
		\Delta  = - \frac{1}{16} s^4 ( a_1^2 - 8 s + 4 ms ) ( a_1^2 + 8 s + 4 m s).
	\end{align} 
$\Delta$ vanishes at order 4 at ${s}=0$, while $c_4$ and $c_6$ are non-zero there. It follows from  Tate's algorithm that the geometric generic fiber over $V(s)$ is of type I$_4$. Since $b_2$ restricted to $V(s)$ is a perfect square, it follows that the generic fiber over $V(s)$ is type I$_4^\text{s}$---this is the main difference between an SO(6)-model and an SO(5)-model.  
 As is clear from the discriminant, we should impose $m\neq \pm 2$ to avoid introducing a type I$_2$ fiber over $V(a_1)$. 
More specifically, we impose $0 \ne m \ne \pm{}2 $ to avoid introducing new sections that will modify the Mordell--Weil group.

\subsection{Crepant resolutions}
\label{SO6res}

There are two isomorphic resolutions connected  to each other by an Atiyah flop, induced by the inverse map of the Mordell--Weil group:
	\begin{align}
			\begin{array}{c}
			\begin{tikzpicture}
				\node(E0) at (0,0) {$ Y_0$};
				\node(E1) at (2.5,0) {$ Y_1$};
				\node(E3) at (5.5,0) {$ Y_2$};
				\node(E2) at (7.5,1.5) {$ Y^-$};
								\node(E2') at (7.5,-1.5) {$ Y^+$};
				\draw[big arrow] (E1) -- node[above,midway]{ $(x,y,s|e_1)$} (E0);
				\draw[big arrow] (E2) -- node[above,midway,sloped]{ $(y,e_1|e_2) $} (E1);
								\draw[big arrow] (E2') -- node[below,midway,sloped]{  $(y+a_1 x,e_1|e_2) $} (E1);
								\draw[<->, dashed] (E2) -- node[right]{ flop} (E2');
			\end{tikzpicture}
		\end{array}
	\end{align}
We first consider the resolution  $ Y^-$. 
The proper transform of $Y_0$ is 
	\begin{align}
		Y^- ~: ~y(e_2 y+ a_1 x ) z= e_1 x (x^2 + ms xz + s^2z^2),~~[ e_2 e_1 x: e_2^2 e_1 y:z],[x:e_2 y:s] ,[ y : e_1]. 
	\end{align}
The fibral divisors are:
	\begin{align}
		D_0 :&\quad s= y(e_2 y + a_1 x) z- e_1 x^3 =0,\quad\quad &&[ e_2 e_1 x: e_2^2 e_1 y:z] ,[ x : e_2 y:0],[y:e_1]\\
		D_1:&\quad e_1 = e_2 y + a_1 x = 0, \quad\quad &&[0:0:z],[x:e_2y:s],[y:0]\\
			D_2 :&\quad e_2 = x =0,\quad\quad &&[ 0:0:z],[0:0:s],[y:e_1] \\
		\label{eqn:SO6C3} D_3 :&\quad e_2=a_1 yz - e_1 (x^2 + msxz + s^2z^2 )=0,\quad && [ 0:0:z],[x:0:s],[y:e_1].
	\end{align} 
Over the degeneration locus $V(a_1)$, we find
\begin{align}
	Y^-~:~
\begin{cases}
C_0 \longrightarrow C_0 \\
C_1 \longrightarrow C_{13} \\
C_2\longrightarrow C_2 \\
C_3 \longrightarrow C_{13}+C_3^++C_3^-
\end{cases}~~~~~Y^+~:~
\begin{cases}
C_0 \longrightarrow C_0 \\
C_1 \longrightarrow C_{13}+C_1^++C_1^- \\
C_2\longrightarrow C_2 \\
C_3 \longrightarrow C_{13}
\end{cases}
\end{align}
where
\begin{align}
Y^-~:~
\begin{cases}
\begin{array}{llll}
		C_0 &:~~ s=e_2 y^2 z- e_1 x^3 =0,~~ &&[ e_2 e_1 x: e_2^2 e_1 y:z] ,[ x : e_2 y:0],[y:e_1]\\
		C_{13}&:~~ e_1 = e_2  = 0, \quad\quad &&[0:0:z],[x:0:s],[y:0]\\
			C_2 &:~~ e_2 = x =0,\quad\quad &&[ 0:0:z],[0:0:s],[y:e_1] \\
		C_3 ^\pm &:~~ e_2= x+\frac{1}{2}(-m \pm \sqrt{m^2-4})  sz=0,\quad && [ 0:0:z],[x:0:s],[y:e_1]
		\end{array}
		\end{cases}
	\end{align} 
and we note that $Y^+$ differs from the case of $Y^-$ described above only by the substitution $C_3^\pm \leftrightarrow C_1^\pm$. 
$Y^\pm$ is nonsingular provided $m\neq 2$. The geometry of the singular fibers is displayed in Figure \ref{fig:SO6}.

	\begin{figure}[htb]
	\begin{center}
			$
			\begin{array}{c}
				\begin{tikzpicture}[scale=1.1]
					\draw[fill=black!8!] (0,0) ellipse (3.7cm and 1.6cm);
					\node at (-2.4,-.6) {$V(p_+)$};
					\node at (2.4,-.6) {$V(p_-)$};
					\node at (0,-1) {\LARGE $B$};
					\node at (3,.2) {$V(s)$};
					\draw[ultra thick] (-2.3,1) --++ (4.6,-2);
					\draw[ultra thick] (2.2,1) --++ (-4.3,-2);
					\draw[ultra thick] (-3,0) --++ (6,0);
					\draw[dotted,thick] (-1.10,.55) -- (-1.10,2.3);
					\draw[dotted,thick] (1.9,.9) -- (1.9,2.4);
					\draw[dotted,thick] (-3,0) -- (-3,2.3);
					\draw[dotted,thick] (0,0) -- (0,3.5 );
					\node at (0,5.4) 
						{$ 
							\begin{tikzpicture}
					\node[draw,circle,ultra thick,xshift=200,label={$C_0$}](C0') at (0,0) {};	
					\node[draw,circle,ultra thick,xshift=200,scale=.65](C13) at (0,-1) {2};
			\node[draw,circle,ultra thick,xshift=200,label=left:{$C_3^+$}](C2) at (-1,-1) {};
			\node[draw,circle,ultra thick,xshift=200,label=right:{$C_2$}](C3-) at (1,-1) {};
			\node[draw,circle,ultra thick,xshift=200,label=below:{$C_3^-$}](C3+) at (0,-2) {};
			\node at (6.6,-.67) {$C_{13}$};
			\draw[ultra thick] (C0') -- (C13) -- (C3-);
			\draw[ultra thick] (C2) -- (C13)--(C3+);
		\end{tikzpicture}
						$};
					\node (I4) at (-3,4) 
						{$
							\begin{tikzpicture}[]
								\node[draw,ultra thick,circle,label=above:{$C_0$}](C0) at (0,0) {};
								\node[draw,ultra thick,circle,label=left:{$C_1$}](C1+) at (-1,-1) {};
								\node[draw,ultra thick,circle,label=right:{$C_3$}](C1-) at (1,-1) {};
								\node[draw,ultra thick,circle,label=below:{$C_2$}](C2) at (0,-2) {};;
								\draw[ultra thick] (C0) -- (C1+) -- (C2) -- (C1-) -- (C0);
							\end{tikzpicture}
						$};
					\node at (-1.3,2.5) 
						{$
							\begin{tikzpicture}[scale=2]
								\draw[scale=.5,domain=-1.2:1.2,variable=\x, ultra thick] plot({\x*\x-1,\x*\x*\x-\x-5});
							\end{tikzpicture}
						$};
					\node at (1.7,2.5) 
						{$
							\begin{tikzpicture}[scale=2]
								\draw[scale=.5,domain=-1.2:1.2,variable=\x, ultra thick] plot({\x*\x-1,\x*\x*\x-\x-5});
							\end{tikzpicture}
						$};
				\end{tikzpicture} 
				\\\\
				\begin{array}{|l|c|}
			\hline
				\text{Weierstrass model} &	 y^2 z + a_1 x y z= x^3 + ms x^2 z + s^2 x z^2 \\\hline
			\text{Discriminant} &\Delta=s^4 p_+ p_- =s^4 ((a_1^2 + 4m s)^2 - 64 s^2 )\\\hline
		\text{Matter representation} & \text{adjoint + vector w/ geometric weight}~ (1,0,-1)\\\hline
		\text{Representation multiplicities} & n_{\text{\bf{15}}}= 1 + L^2;~~ n_{\textbf{6}} = 2 L^2 \\
\hline				\text{Euler characteristic} & \frac{12 L }{1 + 2 L } c(TB)\\\hline
				\text{Triple intersections} &
		\begin{array}{c}			
		6 \mathcal F^- = 2 L^2(- 4 \alpha_0^3 + 3\alpha_0 \alpha_1^2  + 6\alpha_0 \alpha_1 \alpha_3+  3\alpha_0 \alpha_3^2 - 4 \alpha_1^3 \\ +3  \alpha_1^2 \alpha_2 - 6\alpha_1^2 \alpha_3 + 6\alpha_1 \alpha_2 \alpha_3 - 4 \alpha_2^3  +3 \alpha_2 \alpha_3^2 -6 \alpha_3^3),   \\ 
		6 \mathcal F^+ - 6 \mathcal F^- = 4 L^2 (\alpha_3 - \alpha_1)^3 \end{array}
		\\\hline	\end{array}
			\end{array}
			$
			\end{center}
		\caption{Summary of SO(6)-model geometry.}
		\label{fig:SO6}
	\end{figure}
\begin{thm}
The fiber over the generic point of $S$, which is Kodaira type I$_4^\text{s}$, degenerates along the cuspidal locus $V(s,a_2)$ of the elliptic fibration to a type I$_0^*$ fiber. We have\footnote{Given a variety $X$ and an ideal $I$, we denote by $\text{Bl}_I X$ the blowup of $X$ centered at $I$. }
	\begin{align}
	\begin{split}
D_0 &\cong \mathbb P_{S}[ \mathscr{O}_S \oplus \mathscr L]\\ 
D_2&\cong  \mathbb{P}_S[\mathscr{O}_S\oplus \mathscr L]\\
D_1 &\cong \mathbb P_{S}[ \mathscr{O}_S \oplus\mathscr{O}_S]\\
 D_3&\cong \text{Bl}_I (\mathbb P_{S}[ \mathscr{O}_S \oplus\mathscr{O}_S]), \quad I={(a_1, x^2 + m s x z + s^2 z}),
\end{split}
	\end{align}
where for D$_3$, the fibers of the $\mathbb{P}^1$-bundle are paramatrized by $[x:s]$ and the ideal $I$ represents $4K^2$ distinct points, namely, two distinct points on $[S]\cdot[a_1]=2K^2$ distinct fibers of $\mathbb P_{S}[ \mathscr{O}_S \oplus\mathscr{O}_S]$. 
\end{thm}
\begin{proof}
First, following the same argument as for  the SO($5$)-model, we have 
	$
		D_0 \cong \mathbb P_{S}[ \mathscr O_S \oplus \mathscr L].
	$
Next, the defining equation for $C_1$ is $e_2 y =- a_1 x$, which gives the parametrization $[0:0:z][x:-a_1 x:s][y:0]$. Since the divisor is defined in the patch  $zy\neq  0$, we see that  $C_1$ is  parametrized by $[ x/z:s]$. We recall that both $x/z$ and $s$ are of  class $2L-E_1$. Thus, 
	$		D_1 \cong \mathbb P_{S}[ \mathscr O_S \oplus \mathscr O_S].
	$
Similarly, projective scaling implies $C_2$ is parametrized by homogeneous coordinates $ [y/s^2:e_1]$, whence 
	$		D_2 \cong \mathbb P_{S}[\mathscr O_S \oplus \mathscr L]$.   
	Finally observe that $C_3$, can be parametrized by 
	\begin{align}
	\label{eqn:C3}
		a_1 y^\prime - e_1 (x^{\prime 2} + ms x^\prime + s^2 ) =0,~~~~~ [ x^\prime:s],[y^\prime:e_1],~~~~~ (x^\prime,y^\prime)  = (x/z,y/z).
	\end{align} 
This can be viewed as the blowup of the trivial projective bundle $S \times \mathbb P^2_{[x^\prime:s]}$ along the ideal $I = (a_1 ,x^{\prime 2} + ms x^\prime + s^2)$. Rescaling $x^\prime, y^\prime$ by the unit $z$, we recover precisely the hypersurface equation $(\ref{eqn:SO6C3})$.
Since $x'$ and $s'$ are of the same class, the equation $x^{\prime 2} + ms x^\prime + s^2)$ gives two distinct points in the $\mathbb{P}^1$ fiber parametrized by $[x':s']$ over $S\cap V(a_1)$. That means that $D_3$ is obtained from a projective bundle $ \mathbb P_{S}[ \mathscr{O}_S \oplus \mathscr L]$ by blowing up two points 
(corresponding to $x^{\prime 2} + ms x^\prime + s^2=0$ on each of the fiber at the intersection of $S$ and $V(a_1)$, that is a total of $2[S][a_1]=4L^2$ points. 
\end{proof}
It follows from this theorem that $D_0^3=D_1^2=D_3^2=8(1-g)=-8 L^2$, and $D_3^3= 8(1-g) -4K^2=-12 L^3$.

\begin{thm}
The intersection numbers between the divisors $D_a$ and their generic fibers $C_a$ give the following intersection matrix:
	\begin{align}
	\label{eqn:SO6mat}
		\left( D_a C_b \right) &= \begin{tabular}{r}
			$\begin{matrix} C_0 & C_1 & C_2 & C_3 \end{matrix}~~~ $
			\\
			$\begin{matrix}D_0\\D_1\\D_2\\D_3\end{matrix} \begin{bmatrix} -2 & 1 &0 & 1 \\1 & -2 & 1 & 0 \\ 0 & 1 &-2 & 1 \\1 & 0 & 1 &-2 \end{bmatrix} $
		\end{tabular}
	\end{align}
More generally,
	\begin{align}
	\label{eqn:SO6matM}
		\phi_{3*}\left( D_a D_b \cdot \phi_3^* M \cap [{Y_{-}}] \right) &=\begin{tabular}{r}
			$\begin{matrix} C_0 & C_1 & C_2 & C_3 \end{matrix}~~~ $
			\\
			$\begin{matrix}D_0\\D_1\\D_2\\D_3\end{matrix} \begin{bmatrix} -2 & 1 &0 & 1 \\1 & -2 & 1 & 0 \\ 0 & 1 &-2 & 1 \\1 & 0 & 1 &-2 \end{bmatrix} $
		\end{tabular} (S \cdot M ) \cap [ B ] 
	\end{align}
where
	\begin{align}
		M \in A_*(B) ~~ \text{and} ~~ \phi_3 = \text{Bl}_{  {(}x,e_2  {)}} \circ \phi_2. 
	\end{align}
\end{thm}
\begin{proof}
	The divisor class group of $X_2$ is not generated by the classes $H,L,E_1,E_2$. To circumvent this complication, we therefore blow up along $I(D_2) = (x,e_2)$ and factor out a single copy of the exceptional divisor $E_3$; we denote the proper transform of this blowup by $Y_-'$. The second equation above is then a direct pushforward computation with \cite{Euler}
	\begin{align}
	\begin{split}
		D_0 &= 2 L - E_1\\
		D_1 &= E_1 - E_2\\
		 D_2 &= E_3\\
		  D_3 &= E_2 - E_3\\
		  [{Y_-'}]& =( 3 H + 6 L  - 2 E_1 - E_2 - E_3 ).
	\end{split}
	\end{align}
Note that the equation (\ref{eqn:SO6mat}) is obtained from (\ref{eqn:SO6matM}) by choosing $M$ to be the generic point of the divisor $S$ of $B$.
\end{proof}
\noindent The lower right $3\times 3$ block of the intersection matrix (\ref{eqn:SO6mat}) determines (minus) the Cartan matrix of A$_3$:
	\begin{align}
	\label{eqn:SO6metric}
D_a  C_b 	&=\begin{tabular}{l}
				 \  ~ $C_1 $\hspace{.2cm}    $C_2$ \hspace{.1cm}     $C_3$
				 \\
				$
				 \begin{bmatrix}
				 -2 & 1 & 0 \\ 1 & -2 & 1 \\ 0 & 1 & -2 
				 \end{bmatrix}
				 \begin{matrix}
				  D_1\\ 
				  D_2\\ 
				  D_3
				 \end{matrix} %
$
				 \end{tabular}.
	\end{align}

\subsection{Matter representation}
We now the matrix of weight vectors one obtains from the curves over the codimension two degeneration locus $V(t,a_1)$:
\begin{align}
\label{eqn:E-}
			\begin{pmatrix}w_a(C_b)\end{pmatrix}_{Y^-}&=	{ \begin{tabular}{l}
				 \   $C_0 $\hspace{.3cm}    $C_{13}$ \hspace{.1cm}    $C_2$ \hspace{.1cm}     $C_3^+$ \hspace{.1cm}    $C_3^-$
				 \\
				$
				 \begin{bmatrix}
				  -2 & 1 &0 &0  & \ 0 \\ 
				  1 & -2 & 1 & 1 &\   1\\ 
				  0 & 1 & -2 & 0 &\   0\\ 
				  1 &0  & 1 & -1 & -1
				 \end{bmatrix}
				 \begin{matrix}
				  D_0 \\ 
				  D_1\\ 
				  D_2\\ 
				  D_3
				 \end{matrix} %
$
				 \end{tabular}}\\
			\begin{pmatrix} w_a(C_b) \end{pmatrix}_{Y^+}	&={ \begin{tabular}{l}
				 \   $C_0 $\hspace{.3cm}    $C_{1}^+$ \hspace{.1cm}    $C_1^-$ \hspace{.1cm}     $C_2$ \hspace{.1cm}    $C_{13}$
				 \\
				$
				 \begin{bmatrix}
				  -2 & \ 0 &\  0 &\  0  & \ 1 \\ 
				 \   1 & -1 & -1 &\  1 &\   0\\ 
				 \   0 & \  0 &\  0 & -2 &\   1\\ 
				  \  1 &\  1  &\  1 &\  1 & -2
				 \end{bmatrix}
				 \begin{matrix}
				  D_0 \\ 
				  D_1\\ 
				  D_2\\ 
				  D_3
				 \end{matrix} %
$
				 \end{tabular}}
\end{align}
We use the above intersection matrices to determine the matter representation. (In fact, only one of the two is required for this purpose, so we will focus on $Y^-$.) Deleting the affine node, we obtain the following matrix of weight vectors:
	\begin{align}
	\begin{pmatrix} w_a(C_b) \end{pmatrix} &=\begin{bmatrix} \varpi_1 & \varpi_2 & \varpi_3  & \varpi_4 \end{bmatrix} = 		\begin{tabular}{r} $\begin{matrix} C_{13} & C_2 & C_3^+ & C_3^- \end{matrix}~~~~~~ $  \\ $ \begin{bmatrix} 
			-2 & 1 & 1 & 1  \\ 1  &-2 &  0 & 0 \\ 0 & 1 & -1& -1
		\end{bmatrix} \begin{matrix} D_1 \\ D_2 \\ D_3\end{matrix}$ \end{tabular}.
	\end{align}
The new weight is third (and identically, the fourth) column of the above matrix, namely $\varpi_3= (1,0,-1)$. 

The saturation of $\{ \varpi_3\}$ is the $\mathbf{6}$ of SO(6), hence the matter supported on $V(a_1)$ is  in the vector representation $V_{\text{D}_3} = \bigwedge^2_{\text{A}_3}$.

\subsection{Euler characteristic}
	We now compute the generating function for the Euler characteristic, which is the same for both resolutions. 
	
\begin{thm}
The generating function for the Euler characteristic of a SO(6)-model is,
	\begin{align}
		\phi_* c(Y^\pm ) &= \frac{12 L t}{1 + 2 L t } c_t(TB).
	\end{align}
	where $c_t(TB)$ is the Chern polynomial of $B$ and the coefficient of $t^n$ gives the Euler characteristic of a SO(6)-model over a base of dimension $n$. 
	\end{thm}
	
\begin{table}	
	\begin{center}
	
	\begin{tabular}{|c|c|c|}
\hline
	$\dim B$ & Euler characteristic  & Calabi-Yau case\\
	\hline 
	$1$ & $12 L$  & $12 c_1$ \\
	\hline
	$2$ & $12 ( c_1 - 2L ) L$ & $-12 c_1^2$ \\
	\hline 
	$3$ & $12 ( 4L^2 - 2L c_1 + c_2 ) L $ & $ 12 c_1 (2 c_1^2 + c_2) $\\
	\hline 
	$4$ & $12 ( - 8 L^3 + 4 L^2 c_1 - 2 L c_2 + c_3 ) L$ & $ 12 c_1 ( - 4 c_1^3 - 2 c_1 c_2 + c_3)$ \\
	\hline 
	\end{tabular}
	\end{center}
	\caption{Euler characteristic of SO(6)-model for bases of dimension up to $4$. 
	The $i$th Chern class of the base is  denoted  $c_i$. 
	The Calabi-Yau cases are obtained by imposing $L=c_1$. }
	\label{Table:EulerSO6}
\end{table}

\subsection{Triple intersection numbers}

\begin{thm}
The triple intersection polynomial for the crepant resolution $Y^-$ is:
	\begin{align}
	\begin{split}
	\label{eqn:SO6triple}
6		\mathcal F^- = \left( \sum \alpha_i D_i \right)^3 = &\  2 L^2(- 4 \alpha_0^3 + 3\alpha_0 \alpha_1^2 + 3\alpha_0 \alpha_3^2+ 6\alpha_0 \alpha_1 \alpha_3)\\
		&+2L^2( - 4 \alpha_1^3- 4 \alpha_2^3-6 \alpha_3^3 + 3 (\alpha_1^2 + \alpha_3^2+2 \alpha_1 \alpha_3) \alpha_2 - 6\alpha_1^2 \alpha_3 ).
	\end{split}
	\end{align}
	The triple intersection polynomial for $Y^+$ is 
	\begin{align}
	\begin{split}
	\label{eqn:SO6tripleFlop}
6		\mathcal F^+ = \left( \sum \alpha_i D_i \right)^3 =\vspace{.5cm}   &\ 2 L^2(
- 4 \alpha_0^3 + 3\alpha_0 \alpha_1^2 + 3\alpha_0 \alpha_3^2+ 6\alpha_0 \alpha_1 \alpha_3)\\
&+ 2L^2(- 4 \alpha_3^3- 4 \alpha_2^3   -6\alpha_1^3+3 (\alpha_1^2+\alpha_3^2 + 2 \alpha_1 \alpha_3) \alpha_2    -6 \alpha_3^2 \alpha_1  )	.
	\end{split}
	\end{align}	 
	\end{thm}
\noindent Observe that $\mathcal F^\pm$ are related to each other by the involution $\alpha_1\leftrightarrow \alpha_3$, and the difference between the two polynomials is: 
	\begin{equation}
	6\mathcal{F}^--6 \mathcal{F}^+=4 L^2 (\alpha_1-\alpha_3)^3.
	\end{equation}

\subsection{Counting 5D matter multiplets}
The element $\phi$ of the Cartan subalgebra $\mathfrak{h}$ is expressed in a basis of simple coroots. 
The weights $\varpi$ and the roots $\alpha$ are expressed in the basis canonically dual to the basis of simple coroots, namely the basis of fundamental weights.\footnote{This basis is also referred to as the \emph{Dynkin basis} in some references.}  

\subsection{Matter representation}
\begin{table}[htb]
\begin{center}
\begin{tabular}{|c|c|}
\hline 
$\mathbf{6}$  & $\mathbf{15}$ \\ 
\hline 
$
\begin{array}{c}
\boxed{\  0,\   1,\ 0}\\
\\
\boxed{\  1,\ -1,\ 1}\\
\\
 \boxed{\  1,\   0,\ -1}\\
\\
\boxed{\ -1,\ 0, \  1}\\
\\
\boxed{\ -1,\ 1,\ -1}\\
\\
\boxed{ \ 0,\ -1 ,\ 0 }\\ 
\\
\end{array}
$
& 
$
\begin{array}{c}\\
\boxed{\ 1,\ 0,\  1}\\
\\
\boxed{\ -1,\ 1,\ 1}\quad  \boxed{\ 1, \ 1,\ -1}\\
\\
\boxed{\  -1,\ 2, \  -1}\quad   \boxed{\  0,\ -1, \ 2} \quad  \boxed{\ 2, \ -1, \  0} \\
\\
\boxed{\  0,\ 0, \  0}\quad   \boxed{\  0,\ 0, \  0}\quad \boxed{\  0,\ 0, \  0}\\
\\
\boxed{\  -2,\ 1, \  0}\quad   \boxed{\  0,\ 1, \ -2} \quad  \boxed{\ 1, \ -2, \  1} \\
\\
\boxed{\ -1,\ -1,\ 1}\quad  \boxed{\ 1, \ -1,\ -1}\\
\\
\boxed{\ -1,\ 0,\ - 1}\\
\\
\end{array}
$
\\
\hline 
\end{tabular}
\end{center}
\caption{Weights of the representations $\mathbf{6}$, $\mathbf{15}$ of A$_3$ expanded in the basis of fundamental weights. \label{Table:RepA3} }  
\end{table}

The relevant part of the prepotential  is 
\begin{equation}
\mathcal{F}_{\text{IMS}}= \frac{1}{12}\left({\sum_{\alpha} \Big|{(\phi, \alpha)}\Big|^3 - \sum_{\mathbf{R}} n_{\mathbf{R}} \sum_{\varpi}  \Big|(\phi, \varpi)\Big|^3}\right).
\end{equation}

The open fundamental Weyl chamber is the subset of elements of the Cartan subalgebra $\mathfrak{h}$  with positive intersection with all the simple roots. The simple roots of $\text{A}_3$ are:
\begin{equation}
\boxed{2, -1,  0 } \quad \boxed{-1,   2 ,  -1}\quad \boxed{0, -1 , 2 }.
\end{equation}
It follows that the open fundamental Weyl chamber is the cone of $\mathfrak{h}$ defined by:
\begin{equation}
2\phi_1 -\phi_2 >0 \quad -\phi_1 + 2 \phi_2 -\phi_3 >0 \quad -\phi_2 + 2 \phi_3>0.
\end{equation}
\begin{rem}
A weight defines a hyperplane through the origin that intersect the open fundamental Weyl chamber if and only if  both $w$ and $-w$ are not dominant weights.
\end{rem} 
For the SO(6)-model, the representations that we consider are the adjoint and the vector representation. 
By definition, the hyperplanes defined by the  roots  are the walls of the fundamental Weyl chamber.  
The only  non-dominant weight of the vector representation is up to a sign 
  $\boxed{1, 0,  -1}$ with hyperplane $\phi_1-\phi_3=0$. 
It follows thatthe hyperplane arrangement  I$(\text{A}_3, \textbf{6})$ has  two chambers labeled by the sign of the linear form 
\begin{equation}
-\phi_1+\phi_3. 
\end{equation}

\begin{prop}
The prepotential for 
$\mathfrak{so}_6$ with $n_{\textbf{15}}$ matter multiplets transforming in the adjoint representation and $n_\textbf{6}$ matter multiplets transforming in the vector representation can be seen to depend on two phases, corresponding to the sign of the linear form $-\phi_1+\phi_3$. 
The corresponding prepotentials are:\\
\begin{subequations}
\begin{align}
& \begin{aligned}
6 \mathcal{F}^+_{\mathrm{IMS}}=& ~
(8 - 8 n_{\textbf{15}} - 2 n_\textbf{5}) \phi_1^3 + (8 - 8 n_{\textbf{15}}) \phi_3^3+  (8 -    8 n_{\textbf{15}}) \phi_2^3 +3 n_\textbf{6} (\phi_1+\phi_3)^2 \phi_2 \\
&  -  6 n_\textbf{6} \phi_1 \phi_3^2 + 
(-6 + 6 n_{\textbf{15}} - 3 n_\textbf{6}) (\phi_1 +\phi_3)\phi_2^2
\end{aligned}
\\\nonumber \\
&
\begin{aligned}
6 \mathcal{F}^-_{\mathrm{IMS}}=& ~
(8 - 8 n_{\textbf{15}} - 2 n_\textbf{6}) \phi_3^3+ (8 - 8 n_{\textbf{15}}) \phi_1^3  +  (8 -    8 n_{\textbf{15}}) \phi_2^3+3 n_\textbf{6} (\phi_1+\phi_3)^2 \phi_2 \\
&  -  6 n_\textbf{6} \phi_1^2 \phi_3 + 
(-6 + 6 n_{\textbf{15}} - 3 n_\textbf{6}) (\phi_1 +\phi_3)\phi_2^2.
\end{aligned}
\end{align}
\end{subequations}\\
where the $\pm{}$ superscript appearing in the symbol $ \mathcal{F}^\pm_{\mathrm{IMS}}$ is correlated with the $\mathop{sign}(-\phi_1+\phi_3)=\pm$. The prepotentials $ \mathcal{F}^\pm_{\mathrm{IMS}}$  are related to each other by the transposition  $\phi_1 \leftrightarrow \phi_3$. 
\end{prop}

\begin{rem}
The difference between the two $ \mathcal{F}^\pm_{\mathrm{IMS}}$ is proportional to the cube of the linear form:
\begin{equation}
6\mathcal{F}^-_{\mathrm{IMS}}-6\mathcal{F}^+_{\mathrm{IMS}}
=-2 n_\textbf{6} (-\phi_1 +\phi_3)^3
\end{equation}
\end{rem}

\begin{rem}
The prepotential $\mathcal{F}^\pm_{\mathrm{IMS}}$ have monomials of the type $\phi_1 \phi^2_2$  and $\phi_3 \phi^2_2$ that are not present in 
$\mathcal{F}^+$ nor $\mathcal{F}^-$.  Thus, matching the prepotential and the triple intersection numbers will force 
their coefficients to be zero, namely $n_\textbf{6}=2n_{\bf 15}-2$. 
\end{rem}
\begin{prop}
Under the identification $\phi_i = \alpha_i$, the prepotential $\mathcal{F}^\pm_{\mathrm{IMS}}$ matches the triple intersection form $\mathcal{F}^\pm|_{  {\alpha}_0=0}$ if and only if 
\begin{equation}
n_\textbf{6}=2n_{\textbf{15}}-2=2L^2, \quad n_\textbf{15}=1+L^2. 
\end{equation}
If we impose the Calabi-Yau condition and assume the base is a surface, $1+K^2$ is the genus of the curve $V(s)$ of class $2L$ in the base and we have: 
\begin{equation}
n_\textbf{6}=2g-2=2K^2, \quad n_\textbf{15}=g=1+K^2. 
\end{equation}
In particular, the  genus cannot be zero. 
\end{prop}
\begin{proof}
We can compare the prepotential $\mathcal{F}^\pm_{\mathrm{IMS}}$ with the triple intersection form $\mathcal{F}^\pm$  after setting $\alpha_0=0$.  
To have a match of the types of monomials present in the potential, we have to eliminate the coefficients of the terms $\phi_1 \phi_2^2$ and $\phi_3 \phi_2^2$  in  $\mathcal{F}^\pm_{\mathrm{IMS}}$.
This condition imposes $n_\textbf{6}=2n_\textbf{15}-2$. We then get a  perfect match 
$\mathcal{F}^\pm_{\mathrm{IMS}}=\left. \mathcal{F}^\pm\right|_{\alpha_0=0}$  by imposing $n_\textbf{15}=1+L^2$. 
\end{proof}

\section{Hodge numbers}
\label{sec:hodge}

In this section, we compute the Hodge numbers of the SO(3), SO(5), and SO(6)-models assuming that the base $B$ is a smooth rational surface and the resolved elliptic fibration is a Calabi-Yau threefold. 
As usual for these models, the Mordell--Weil group has trivial rank and  torsion $\mathbb{Z}/2\mathbb{Z}$.

\begin{thm}\label{Thm:HodgeSO}
Let $B$ be a smooth compact rational surface with canonical class $K$. Let $Y\longrightarrow B$ be the crepant resolution of an SO(3), SO(5), or SO(6)-model over $B$. 
If $Y$ is a Calabi-Yau threefold then the non-zero  Hodge numbers of $Y$ are $h^{1,1}=h^{3,3}=h^{3,0}=h^{0,3}=1$, and $h^{1,1}$, $h^{2,1}=h^{1,2}=\mathrm{dim}\  H^1 (\Omega^2_Y)$ given by Table \ref{Table:TopologyCY3}.
The Euler characteristic of a Calabi-Yau threefold is 
\begin{equation}
\chi(Y)=2(h^{1,1}-h^{2,1}).
\end{equation}
\end{thm}

\begin{table}[htb]
\begin{center}
\begin{tabular}{|c|c| c |c| c |c|}
\hline 
Model & $\chi(Y)$ & $g$ &
$h^{1,1}(Y)$ &  $h^{2,1}(Y)$\\
\hline 
SO(3) & $-36 K^2$ & $1+6K^2$& 
$12-K^2$ & $12+17 K^2$ \\
\hline 
SO(5) & $-20 K^2$ & $1+K^2$ & 
$13-K^2$  & $13+9 K^2$\\
\hline 
SO(6) & $-12 K^2$ & $1+K^2$ & 
$14-K^2$  & $14+5K^2$\\
\hline
\end{tabular}
\begin{tabular}{|c|c|}
\hline
 $n_{\textbf{adj}}$  & $n_{V}$  \\
\hline
$g$ & $0$  \\
\hline
$g$ &  $3(g-1)$  \\
\hline
$g$ & $2(g-1)$\\
\hline
\end{tabular}
\end{center}
\caption{\label{Table:TopologyCY3} Euler characteristic and Hodge numbers for SO(3), SO(5), and SO(6)-models in the case of a Calabi-Yau threefold over a compact rational surface $B$ of canonical class $K$. 
The divisor $S$ is curve of genus $g$. 
The number of multiplets transforming in the adjoint  and  vector representations  are respectively $n_{\textbf{adj}}$  and  $n_{V}$.
}
\end{table}

\section{Application to F/M theory compactifications}
\label{sec:stringy}

We explore a particular application of our geometric results to compactifications of F/M-theory on elliptically-fibered Calabi-Yau threefolds. In particular, we determine the 6D gauge theoretic descriptions associated to the SO($n$)-models described in this paper in the special case that they are elliptically fibered Calabi-Yau threefolds. The low energy effective description of F-theory compactified on an elliptically fibered Calabi-Yau threefold is 6D $(1,0)$ supergravity. However, supergravity theories in 6D have gravitational, gauge, and mixed anomalies at one loop due to the presence of chiral matter, and therefore anomaly cancellation places strong constraints on the matter spectrum. Determining the possible 6D matter content consistent with anomaly cancellation is therefore a task of primary importance for stringy compactifications of this sort.

We begin this section by reviewing relevant aspects of supergravity theories with 8 supercharges in 5D and 6D. We then determine the number of  hypermultiplets charged in a given representation by solving the anomaly cancellation conditions and checking that they match the results of the 5D computations described earlier. The match between the 6D and 5D matter spectra is essentially due to the fact that F-theory compactified on an elliptically fibered Calabi-Yau threefold times a circle is dual to M-theory compactified on the same threefold, which implies that the gauge theory sector of the 6D $(1,0)$ theory compactified on a circle admits a description as a 5D $\mathcal N=1$ theory.

\subsection{6D $\mathcal{N}=(1,0)$ supergravity }
We collect some useful facts about 6D theories \cite{Park,Sadov:1996zm}. Six-dimensional (gauged) supergravity with 8 supercharges has SU(2) R-symmetry. 
The fermions of the theory can be formulated as symplectic Majorana--Weyl spinors, which transform in the fundamental representation of the SU(2) R-symmetry group. 
There are four types of massless on-shell supermultiplets: a graviton multiplet, $n_{\text T}$ tensor multiplets, $n_{\text V}^{(6)}$ vector multiplets characterized by a choice of gauge group, and $n_{\text H}$ hypermultiplets transforming in a representation of the gauge group.

We will call an antisymmetric $p$-form with self-dual (resp. anti-self-dual)  field strength a self-dual (resp. anti-self-dual) $p$-tensor field. In addition, we will simply refer $2$-tensor fields as `tensors'. 
The graviton multiplet contains an anti-self-dual tensor, while each tensor multiplet includes a self-dual tensor. The self-duality properties of these tensor fields cannot be derived consistently from a known action principle,  and consequently 6D $\mathcal{N}=(1,0)$ supergravity does not at present have a conventional Lagrangian formulation for $n_{\text T}>1$.

The tensor multiplet scalars parametrize the homogeneous symmetric space $\text{SO}(1,n_{\text T})/\text{SO}(n_{\text T})$ \cite{Romans:1986er}. The quaternionic scalars parametrize locally a non-compact quaternionic-K\"ahler manifold. 
The chiral tensor multiplets can induce local anomalies that have to be cancelled.

\begin{table}[htb]
\begin{center}
\begin{tabular}{|c|c|}
\hline
Multiplet & Fields \\
\hline 
Graviton &  $(g_{\mu\nu}, B^-_{\mu\nu},\psi^-_\mu)$\\\hline
Vector &  $(A_\mu, \lambda^-)$\\\hline
Tensor & $(B^+_{\mu\nu}, \phi, \chi^+)$\\\hline
Hyper &  $(q, \zeta^+)$\\\hline
\end{tabular}
\end{center}
\caption{Supermultiplets of $\mathcal{N}=(1,0)$ six-dimensional supergravity.
The indices $\mu$ and $\nu$ refer to the six-dimensional spacetime coordinates. 
The tensor $g_{\mu\nu}$ is the metric of the six0dimensional spacetime.  The fields $\psi^-_\mu,\lambda^-, \chi^+,  \zeta^+$  are symplectic Majorana--Weyl spinors.
The field $\psi^-_\mu$ is the gravitino.  
 The chirality of fermions is indicated by the $\pm$ superscript. The tensor $B^+_{\mu\nu}$ is a two-form with self-dual field strength, while  $B^-_{\mu\nu}$  is a two-form with anti-self-dual field-strength.
  The scalar field $\phi$ is a  pseudo-real field. The hypermultiplet scalar $q$ is a quaternion composed of four pseudo-real fields. 
 \label{Table:6DMatter}}
\end{table}

\subsection{5D $\mathcal{N}=1$ supergravity }
Five-dimensional (Yang-Mills-Einstein) supergravity with 8 supercharges has SU(2) R-symmetry. 
All spinors are symplectic Majorana spinors and transform in the fundamental representation of the SU(2) R-symmetry group.  
There are three types of massless on-shell supermultiplets: a graviton multiplet, $n_{\text V}^{(5)}$ vector multiplets, and $n_{\text H}^{(5)}$ hypermultiplets. The graviton multiplet contains a vector field called the graviphoton.

The scalar fields of the hypermultiplets are called hyperscalars, and transform in the fundamental representation of SU(2). 
Each hyperscalar is a complex doublet, giving altogether four real fields. The hyperscalars  collectively 
parametrize a quaternionic-K\"ahler manifold of real dimension $4n_{\text H}^{(5)}$. The vector multiplet scalars $\phi$ parametrize a real $n_{\text V}^{(5)}$ dimensional manifold called a {\em very special real manifold} which can described in terms of \emph{very special coordinates} as a hypersurface ${\mathscr{F}}=1$ of an affine real space of dimension  $n_{\text V}^{(5)}+1$. 

The dynamics of the gravity and vector fields at the two-derivative level are completely determined by a real cubic polynomial ${\mathscr{F}}$ whose coefficients are the Chern-Simons couplings appearing in the 5D action. In particular, the cubic potential ${\mathscr{F}}$ determines both the matrix of gauge couplings and the metric on the real manifold parametrized by the vector multiplet scalars. 

\begin{table}[htb]
\begin{center}
\begin{tabular}{|c|c|}
\hline
Multiplet & Fields \\
\hline 
Graviton &  $(g_{\mu\nu}, A_{\mu},\psi_\mu)$\\\hline
Vector  &  $(A_\mu,\phi, \lambda)$\\\hline
Hyper &  $(q, \zeta)$  \\\hline
\end{tabular}
\end{center}
\caption{Supermultiplets for $\mathcal{N}=1$ five-dimensional supergravity.
The indices $\mu$ and $\nu$ refer to the five-dimensional spacetime coordinates. 
The tensor $g_{\mu\nu}$ is the metric of the five-dimensional spacetime.  The fields $\psi_\mu,\lambda,\zeta$  are symplectic Majorana spinors.
The field $\psi_\mu$ is the gravitino and $A_\mu$ is the graviphoton.   
The hyperscalar is a quaternion. 
   \label{Table:5DMatter}}
\end{table}

\subsubsection*{Kaluza-Klein reduction from 6D to 5D}
If we compactify F-theory on $Y$ to 6D and then further compactify on a circle $S^1$, we anticipate the effective description will include the same field content as the compactification of M-theory on $Y$ to 5D. 
This is summarized in  Table \ref{Table:FM}.

In this section we summarize the Kaluza Klein (KK) reduction of 6D $\mathcal N=(1,0)$ on a circle and the relation of this theory to 5D $\mathcal N =1$ supergravity. To facilitate a comparison between the dual F/M theory compactifications on a smooth threefold, we study these 5D theories on the Coulomb branch, where the non-Cartan vector fields and charged hypermultiplets acquire masses due to spontaneous gauge symmetry breaking and are subsequently integrated out.

It is possible to cast the massless fields of the 6D KK reduction in the canonical 5D framework.  First, the neutral 6D hypermultiplets descend to 5D hypermultiplets,
	\begin{align}
		n_{\text H}^{(5)} = n_{\text H}^{0},	
	\end{align}
where we use $n_{\text H}^0$ to denote neutral hypermultiplets in 6D. 

We next turn our attention to the vector fields. Dimensional reduction of the 6D tensor fields produces both 5D tensor and vector fields. However, the 6D self-duality condition also descends to a 5D constraint that can be imposed at the level of the action to dualize all of the 5D  massless tensors to vectors. Thus we are free to assume that on top of the graviton multiplet, the  field content of  our 5D KK theory is comprised solely of vector multiplets and hypermultiplets. The 5D KK reduced vector fields include the graviphoton, $n_{\text T} +1$ tensors, and $n_{\text V}^{(6)}$ vectors, making for a total of $n_{\text T}+n_{\text V}^{(6)}+2$ vector fields.\footnote{We emphasize here that when a comparison is made between the 6D theory compactified on a circle and the 5D theory on the Coulomb branch, the number $n_{\text V}^{(6)}$ is equal to the number of uncharged 6D vectors.} Accounting for the fact that one of the vector fields must belong to the graviton multiplet, the total number of 5D vector multiplets is given by \cite{Cadavid:1995bk}:
\begin{equation}
n_{\text V}^{(5)}=n_{\text T}+n_{\text V}^{(6)}+1.
\end{equation}
The next step in our analysis is to understand the geometric origin of the field theoretic data via the F/M theory compactifications. 
\begin{table}
\begin{center}
\begin{tabular}{|c|}
\hline 
 6D $\mathcal{N}=(1,0)$ sugra on $\mathbb R^{1,4} \times S^1$ \\
 $\downarrow$\\
 5D $\mathcal{N}=1$ sugra on $\mathbb R^{1,4}$ \\
 \hline 
  $n_{\text V}^{(5)}=n_{\text V}^{(6)}+n_{\text T}+1$\\
$n_{\text H}^{(5)} = n^{0}_H$\\
 \hline 
 \end{tabular}
 \end{center}
 \caption{Identification between multiplets from KK reduction of 6D $\mathcal N=(1,0)$ and 5D $\mathcal N =1$ supergravity multiplets on the Coulomb branch.}
\end{table}

  \subsection{M-theory on a Calabi-Yau threefold}
  
 \begin{table}[htb]
 \begin{center}
\begin{tabular}{|c|c|c|}
\hline 
F-theory on $Y$ & M-theory on $Y$ &  F-theory on $Y\times S^1$  \\
$\downarrow$ & $\downarrow$& $\downarrow$\\
 D=6 $\mathcal{N}=(1,0)$ sugra  &  D=5 $\mathcal{N}=1$ sugra &  D=5 $\mathcal{N}=1$ sugra \\
 \hline 
$n_{\text V}^{(6)}+n_{\text T}=h^{1,1}(Y)-2$ & $n_{\text V}^{(5)}=h^{1,1}(Y)-1$ &   $n_{\text V}^{(5)}=h^{1,1}(Y)-1$\\
 $n_{\text H}^0=h^{2,1}(Y)+1$& $n_{\text H}^0=h^{2,1}(Y)+1$& $n_{\text H}^0=h^{2,1}(Y)+1$ \\
 $n_{\text T}=h^{1,1}(B)-1$ & &   \\
 \hline 
 \end{tabular}
 \end{center}
 \caption{\label{Table:FM} Compactification of F-theory and M-theory on a smooth Calabi-Yau threefold $Y$. We assume that all two-forms in the five-dimensional theory are dualized to vector fields. The numbers of neutral hypermultiplets are the same in 6D and 5D. 
 By contrast, the number of 5D vector multiplets is $n_{\text V}^{(5)}=n_{\text V}^{(6)}+n_{\text T}+1$, where we emphasize that the 5D theory is on the Coulomb branch.} 
 \end{table}

Compactification of 11D supergravity on a resolved Calabi-Yau threefold $Y$ leads to 
 5D $\mathcal{N}=1$ supergravity on the Coulomb branch, coupled to $n_{\text V}^{(5)}=h^{1,1}(Y)-1$ vector multiplets and $n_{\text H}=h^{1,2}(Y)+1$ hypermultiplets \cite{Cadavid:1995bk}. 
 The Coulomb branch of the 5D gauge sector is parametrized by the $n_{\text V}^{(5)}$ vector multiplet scalars, and  corresponds to the extended K\"ahler cone of $Y$ restricted to the unit volume locus.\footnote{There is an additional K\"ahler modulus controlling the overall volume of $Y$ which belongs to the universal hypermultiplet and is thus not counted among the 5D Coulomb branch parameters.
In this case, we expand a K\"ahler class in a basis of $h^{1,1}(Y)-1$ irreducible effective divisors identified with the coroots of the affine Dynkin diagram $\widetilde{\mathfrak{g}}^t$.} The Chern-Simons couplings determining the one-loop quantum corrected prepotential  on the Coulomb branch are identified as the triple intersection numbers of the effective irreducible divisors of $Y$, appearing as coefficients of the triple intersection polynomial.

Given a $G$-model with representation $\textbf{R}$, we consider the triple intersection polynomial $\mathcal F$ of the $G$-model  and the prepotential $\mathcal F_{\text{IMS}}$ of a 5D gauge theory with gauge group $G$ and an undetermined number of hypermultiplets in the representation $\textbf{R}$. We can determine which values of the numbers of charged hypermultiplets are necessary to get a perfect match between $\mathcal F|_{\alpha_0 = 0}$ and $\mathcal F_{\text{IMS}}$.

  \subsection{F-theory  on a Calabi-Yau threefold  
and anomaly  cancellation}
In this section, we prove that the SO(3), SO(5), and SO(6)-models define anomaly free 6D $\mathcal{N}=(1,0)$ supergravity theories. 

\subsubsection{Generalized Green-Schwarz mechanism in F-theory}

The low energy effective description of F-theory compactified on an elliptically fibered Calabi-Yau threefold $Y$ with a base $B$ is six-dimensional $\mathcal{N}=(1,0)$ supergravity coupled to 
$n_{\text T}=h^{1,1}(B)-1$ tensor multiplets, $n_{\text H}^0=h^{2,1}(Y)+1$ neutral hypermultiplets, and  $n_{\text V}^{(6)}$ massless (Cartan) vector multiplets such that  $n_{\text V}^{(6)}+n_{\text T}=h^{1,1}(Y)-2$ \cite{Morrison:1996pp}.
That is: 

\begin{equation}
h^{1,1}(Y)=n_{\text V}^{(6)}+h^{1,1}(B)+1.
\end{equation}
We assume that $Y$ is  a simply-connected elliptically fibered Calabi-Yau threefold with holonomy SU($3$). The restriction on the holonomy is stronger than the condition  that $Y$ has a trivial canonical class, in particular, it  implies that  $h^{1,0}(Y)=h^{2,0}(Y)=0$. Then we also have $h^{2,0}(B)=h^{1,0}(B)=0$ and  the Enriques--Kodaira classification  identifies  $B$ to be a rational surface or an Enriques surface. 
However, the condition that $Y$ is simply connected also requires that $B$ is simply connected and this rules out the Enriques surface \cite[Chap VI]{Barth}. Thus $B$ is a rational surface.\footnote{  A rational surface is a surface birational to the complex projective plane $\mathbb{P}^2$. Any smooth rational surface is   $\mathbb{P}^2$, the Hirzebruch surface $\mathbb{F}_n$ ($n\neq 1$), or derived from them by a finite sequence of blowups \cite[Theorem V.10]{Beauville}. An Enriques surface is a surface with $h^{1,0}=0$, $K_B^2=0$ but $K_B\neq \mathscr{O}_B$. An Enriques surface has $h^{1,1}=10$ , $h^{1,0}=h^{2,0}=0$.} The only non-trivial Hodge number of $B$ is $h^{1,1}(B)$,  its Euler characteristic is $c_2(TB)=2+h^{1,1}(B)$ and its signature $\tau(TB)=2-h^{1,1}(B)$. Moreover, since $h^{0,1}(B)=0$, 
its holomorphic Euler characteristic $\chi(\mathscr{O}_B)=h^{0,0}-h^{0,1}$  is $1$.  Thus, Noether's theorem $(12\chi(\mathscr{O}_B)= K^2+c_2(TB)$)    implies that $h^{1,1}(B)=10 -K^2$. 
The Hodge index theorem states that the intersection of two-forms in $B$ has signature $(1, h^{1,1}(B)-1)$.

Consider a 6D ${\cal N}=(1,0)$ supergravity theory with $n_{\text T}$ tensor multiplets, $n_{\text V}^{(6)}$ vector multiplets, and $n_{\text H}$ hypermultiplets. 
We assume that the gauge group is simple. 
We distinguish $n_{\text H}^0$ neutral hypermultiplets and n$_{\mathbf{R_i}}$ hypermultiplets transforming in the representation $\mathbf{R}_i$ of the gauge group. 
CPT invariance requires that $\textbf{R}_i$ is a quaternionic representation \cite{Schwarz:1995zw}, but if $\bf{R}_i$ is pseudo-real, 
but when $\bf{R}_i$ is pseudo-real, we can have half-hypermultiplets transforming under $\bf{R}_i$, which can give half-integer values for $n_{\bf{R}_i}$. 
Following \cite{GM1}, we count as neutral any hypermultiplet whose  charge is given by the zero weight of a representation. 
We denote by $\dim{\mathbf R}_{i,0}$ the number of zero weights in the representation ${\mathbf R}_{i}$. 
The total number of charged hypermultiplets is then \cite{GM1}
\begin{equation}
n_{\text H}^{\text{ch}}=\sum_i (\dim \mathbf{R}_i -\dim{\mathbf R}_{i,0}) n_{\mathbf{R_i}},
\end{equation}
and the total number of hypermultiplets is  $n_{\text H}=n_{\text H}^0+n_{\text H}^{\text{ch}}$.  
The pure gravitational anomaly  is cancelled by the vanishing of the coefficient of $\mathrm{tr}\  R^4$ in the anomaly polynomial \cite[Footnote 3]{Salam}:
\begin{equation}
\label{eqn:grav}
n_{\text H}-n_{\text V}^{(6)}+29 n_{\text T}-273=0.
\end{equation}
 Using the duality between F-theory on an elliptically fibered Calabi-Yau threefold with base $B$ and type IIB on $B$, Noether's formula implies the following for the number of tensor multiplets \cite{Sadov:1996zm}:
\begin{equation}
n_{\text T}=h^{1,1}(B)-1=9 - K^2.
\end{equation}
If the gauge group is a simple group $G$, the remaining part of the anomaly polynomial is \cite{Schwarz:1995zw}:
 \begin{equation}
\mathcal  I_8= \frac{9-n_{\text T}}{8} (\mathrm{tr} \   R^2)^2+\frac{1}{6}\  X^{(2)} \mathrm{tr}\ R^2-\frac{2}{3} X^{(4)},
 \end{equation}
where 
\begin{align}
X^{(n)}=\mathrm{tr}_{\mathbf{adj}}\  F^n -\sum_{i}n_{\mathbf R_i} \mathrm{tr}_{\mathbf R_i}\  F^n.
\end{align}
Choosing a reference representation $\mathbf F$, we have after some trace identities:
The trace identities for a representation $\mathbf{R}_{i}$ of a simple group $G$ are
\begin{equation}
\tr_{\bf{R}_{i}} F^2=A_{\bf{R}_{i}} \tr_{\bf{F}} F^2 , \quad \tr_{\bf{R}_{i}} F^4=B_{\bf{R}_{i}} \tr_{\bf{F}} F^4+C_{\bf{R}_{i}} (\tr_{\bf{F}} F^2)^2
\end{equation}
with respect to a reference representation $\bf{F}$ for each simple component $G$ of the gauge group.\footnote{We denote this representation by $\bf{F}$ as we have chosen the fundamental representation(s) for convenience. However, any representation can be used as a reference representation.} The coefficients $A_{\bf{R}_{i}}$, $B_{\bf{R}_{i}}$, and $C_{\bf{R}_{i}}$ depends on the gauge groups and are listed in \cite{Erler,Avramis:2005hc,vanRitbergen:1998pn}. We then have 
\begin{align}
X^{(2)}&=\Big(A_{\textbf{adj}}   -\sum_{i}n_{\mathbf R_i} A_{\mathbf R_i}\Big) \mathrm{tr}_{\mathbf F}  F^2\\
X^{(4)}&=\Big(B_{\textbf{adj}}   -\sum_{i}n_{\mathbf R_i} B_{\mathbf R_i}\Big) \mathrm{tr}_{\mathbf F}  F^4 +
\Big(C_{\textbf{adj}}   -\sum_{i}n_{\mathbf R_i} C_{\mathbf R_i}\Big)( \mathrm{tr}_{\mathbf F} F^2)^2
.
\end{align}
  If  $G$  does not have two independent quartic Casimir invariants, we take $B_{\bf{R}_{i}}=0$ \cite{Sadov:1996zm}. 
In a gauge theory with at least two quartic Casimirs, to have a chance to cancel the anomaly, the coefficient of $\mathrm{tr}_{\mathbf F}  F^4$ must vanish: 
\begin{align}
 B_{\textbf{adj}}   -\sum_{i}n_{\mathbf R_i} B_{\textbf{R}_i}=0.
\end{align}
We are then left with:
 \begin{align}
2 \mathcal  I_8 &= \frac{9-n_{\text T}}{4} (\mathrm{tr} \   R^2)^2+\frac{1}{3} \left(A_{\bf{adj}}-\sum_{i}n_{\bf{R}_{i}}A_{\bf{R}_{i}}\right)   ( \mathrm{tr}_{\mathbf F} F^2)\  (\mathrm{tr}\ R^2)-\frac{4}{3} \left(C_{\bf{adj}}-\sum_{i}n_{\bf{R}_{i}}C_{\bf{R}_{i}}\right)  ( \mathrm{tr}_{\mathbf F} F^2)^2.
 \end{align}
If the simple group $G$ is supported on a divisor $S$ and $K$ is the canonical class of the base of the elliptic fibration,
we can  factor $\mathcal{I}_8$ as a perfect square following Sadov's analysis \cite{Sadov:1996zm,Sagnotti:1992qw}: 
 \begin{align}
\mathcal I_8 &= \frac{1}{2}(\frac{1}{2} K\mathrm{tr} \   R^2-\frac{2}{\lambda} S\ \mathrm{tr}_{\mathbf F} F)^2.
 \end{align}
Sadov showed this factorization matches the general expression of $\mathcal  I_8$  if and only the following anomaly cancellation conditions hold  \cite{Sadov:1996zm} (see also \cite{GM1,Park}):
\begin{subequations}\label{eq:AnomalyEqn}
\begin{align}
n_{\text H}-n_{\text V}^{(6)}+29n_{\text T}-273 &=0,\label{Grav.an}\\
n_{\text T}&=9-K^2 , \\
\left(B_{\bf{adj}}-\sum_{i}n_{\bf{R}_{i}}B_{\bf{R}_{i}}\right)& = 0, \label{B.an}\\
\lambda \left(A_{\bf{adj}}-\sum_{i}n_{\bf{R}_{i}}A_{\bf{R}_{i}}\right) & =6  K S, \\
\lambda^2 \left(C_{\bf{adj}}-\sum_{i}n_{\bf{R}_{i}}C_{\bf{R}_{i}}\right) & =-3 S^2.
\end{align}
\end{subequations}
The coefficient  $\lambda$ is a  normalization factor  chosen such that the  smallest topological charge of an embedded SU($2$) instanton in $G$ is one \cite{Kumar:2010ru, Park, Bernard}. This forces $\lambda$ to be the Dynkin index of the fundamental representation of  $G$ as summarized in Table \ref{tb:normalization} \cite{Park}. 

Using adjunction ($KS+S^2=2g-2$), the last two anomaly equations give an expression for the genus of $S$:
\begin{equation}\label{Witten.an}
\lambda \left(A_{\bf{adj}}-\sum_{i}n_{\bf{R}_{i}}A_{\bf{R}_{i}}\right) -2\lambda^2 \left(C_{\bf{adj}}-\sum_{i}n_{\bf{R}_{i}}C_{\bf{R}_{i}}\right) =12  (g-1).
\end{equation}

\begin{table}[htb]
\begin{center}
\begin{tabular}{|c|c|c|c|c|c|c|c|c|c|}
\hline
 $\mathfrak{g}$ & A$_n$ ($n\geq 1$) & B$_n$  ($n\geq 3$) & C$_n$ ($n\geq 2$) & D$_n$  ($n\geq 4$)& E$_8$ & E$_7$ & E$_6$&  F$_4$ & G$_2$ \\
 \hline
 $\lambda$ & $1$ & $2$  & $1$ & $2$ & $60$ & $12$ & $6$ & $6$ & $2$ \\
 \hline  
\end{tabular}
\caption{The normalization factors for each simple gauge algebra. See \cite{Kumar:2010ru}.}
\label{tb:normalization}
\end{center}
\end{table}

\subsubsection{Anomaly cancellations for SO($3$), SO($5$), and SO($6$)-models}
We will now consider the specificity of the SO($3$), SO($5$), and SO($6$)-model. 
Since we use the vector representation of SO($3$), SO($5$), and SO($6$) as the reference representation, $\lambda$ takes the respectively the same value as for SU($2$), B$_n$ and D$_n$.

\begin{table}[hbt]
\begin{center}
\begin{tabular}{|c|c|c|c|c|c|c|c|c|c|}
\hline 
$\mathfrak{g}$ & $\lambda$ & $A_{\textbf{adj}}$ & $B_{\textbf{adj}}$ & $C_{\textbf{adj}}$ & $A_{V}$ & $B_{V}$ & $C_{V}$ & S & $g(S)$\\
\hline 
 $\mathfrak{so}$(3)&  $4$& $1$ & $0$& $1/2$ & $1$ & $0$ &$1/2$ & -4K &$1+6K^2$\\
 \hline
 $\mathfrak{so}$(5)& $2$ & $3$&  $-3$& $3$ & $1$&  $1$&$0$ & -2K&$ 1+ K^2$\\
 \hline 
 $\mathfrak{so}$(6)& $2$ & $4$&$-2$  & $3$ &$1$ & $1$ &$0$ & -2K& $1+ K^2$\\
\hline
\end{tabular}
\caption{Coefficients for the trace identities in the case of  SO($3$), SO($5$), and SO($6$). In all cases the reference representation $\mathbf{F}$ is the vector representation---namely, the $\bf{3}$ of SO($3$), the $\bf{5}$ of SO($5$), and the $\bf{6}$ of SO($6$). 
 \label{table:SONormalization} }
\end{center}
\end{table}
We will need the following trace identities \cite{PVN,Schwarz:1995zw,Sagnotti:1992qw}: 
\begin{align}
\text{SO}(3)\quad\qquad:&\quad  \tr_{\bf{adj}} F^2= \tr_{\bf{vec}} F^2 , \quad \tr_{\bf{adj}} F^4=\frac{1}{2} (\tr_{\bf{vec}} F^2)^2,\\
\text{SO}(n) \quad n\geq 5: &\quad  \tr_{\bf{adj}} F^2=(n-2) \tr_{\bf{vec}} F^2 , \quad \tr_{\bf{adj}} F^4=(n-8) \tr_{\bf{vec}} F^4+3 (\tr_{\bf{vec}} F^2)^2.
\end{align}

	For the $\text{SO}(3)$-model the only representation that we consider  is in the adjoint ($\bf{3}$) and we also use this representation as our reference representation. 
For the $\text{SO}(5)$ and $\text{SO}(6)$-model, we find matter in the adjoint and the vector representations. 
We use the vector representation as the fundamental representation. Thus,  our choice of $\lambda$ follows from the trace identities 
\begin{align}
& B_1=A_1: \quad \tr_{\bf{3}}F^2=4\tr_{\bf{2}} F^2, \\
& B_2=C_2:\quad \tr_{\bf 5} F^2=2  \tr_{\bf 4} F^2, \\
& D_3=A_3: \quad \tr_{\bf 6} F^2=2 \tr_{\bf 4} F^2.
\end{align}

We first ignore the condition for the cancellation of the gravitational anomaly, namely equation \eqref{Grav.an}. 
After fixing the conventions for the trace identities and the coefficient $\lambda$,  we are left with  linear equations that have a unique solution. 
For SO($3$), all the remaining equation gives $n_{\textbf{adj}} =g$. 
For SO($n$) with $n=5,6$, equation \eqref{B.an} gives $n_{V}=(8-n)(n_{\textbf{adj}}-1)$. 
Feeding this in equation  \eqref{Witten.an} gives $n_{\textbf{adj}}=g$ and all the remaining equations are satisfied:
		\begin{align}
		\text{SO}(3): &\quad	n_{\textbf{adj}} = 1 + 6 K^2 = g, \\
\text{SO}(5):&\quad		n_{\textbf{adj}} = 1 + K^2 = g ,~~ n_{V} = 3 K^2=3(g-1), \\ 
\text{SO}(6):&\quad		n_{\textbf{adj}} = 1+ K^2 = g , ~~ n_{V} = 2K^2=2(g-1). 
	\end{align}

We are left with the pure gravitational  anomaly, which requires checking  equation  \eqref{Grav.an}.  Since we have explicit expressions for the number of charged hypermultiplets, this is a straightforward computation.
We recall that:
\begin{align}
n_{\text T}=9-K^2, \quad 
n_{\text H}=n_{\text H}^0 +n_{\text H}^{\text{ch}}, \quad n_{\text H}^0 =h^{2,1}(Y)+1, \quad 
 n^{\text{ch}}_{\text H} =
  \sum_{i} ( \mathrm{dim}\ \mathbf{R}_i-\mathrm{dim} \mathbf{R}_{i,0}) n_{\mathbf R_i}
  \end{align}
For all  adjoint representations, the number of zero weights is the number of simple roots, which is the rank of the Lie algebra. 
The vector representation of SO($6$) does not have any zero weights, but the vector representation of SO($5$) has exactly one zero weight. 

For SO($3$), we have 
\begin{align}
{ n_{\text V}^{(6)}}&=\text{dim}(\textbf{adj})=3, \\
n_{\text H}^0 &=h^{2,1}(Y)+1=13+17K^2,  \\
n_{\text H}^{\text{ch}} &=\big(\mathrm{dim}\ {\mathfrak{so}(3)}-\mathrm{rk}\ \mathfrak{so}(3)\big) { n_{\text V}^{(6)}}=2+12 K^2.
  \end{align}

For SO($5$), we have 
\begin{align}
   n_{\text V}^{(6)}&=\text{dim}(\textbf{adj})=10, \\
n_{\text H}^0 &=h^{2,1}(Y)+1=14+9K^2,  \\
n^{\text{ch}}_{\text H} &=(\mathrm{dim}\ {\mathfrak{so}(5)}-\mathrm{rk}\  \mathfrak{so}(5))n_{\textbf{adj}}+(5-1) n_{\mathbf{5}}=8+20 K^2.
  \end{align}
  In computing $n^{\text{ch}}_{\text H}$, we use $(5-1) n_{\mathbf{5}}$ instead of $5 n_{\mathbf{5}}$ since the representation $\mathbf{5}$  has one zero weight.  
  
For SO($6$), we have 
\begin{align}
{ n_{\text V}^{(6)}}&=\text{dim}(\textbf{adj})=15, \\
n_{\text H}^0 &=h^{2,1}(Y)+1=15+5K^2,  \\
n_{\text H}^{\text{ch}} &=\big(\mathrm{dim}\ {\mathfrak{so}(6)}-\mathrm{rk}\ \mathfrak{so}(6)\big) { n_{\text V}^{(6)}}+6 n_{\mathbf{6}}=12+24 K^2.
  \end{align}
We then see immediately that in all three models, equation \eqref{Grav.an} is also satisfied.

\section*{Acknowledgements}
The authors are grateful to Chris Beasley, 
 James Halverson, Ravi Jagadeesan, Monica Kang, Cody Long, Sabrina Pasterski, and Julian Salazar for discussions. 
P.J. would like to thank the organizers of the 2018 Summer Workshop at the Simons Center
for Geometry and Physics for their hospitality and support during part of
this work.  
M.E. is supported in part by the National Science Foundation (NSF) grant DMS-1406925  and DMS-1701635 ``Elliptic Fibrations and String Theory''.
P.J.  is  supported by the Harvard University Graduate Prize Fellowship. P.J. would like to extend his gratitude to Cumrun Vafa for his tutelage and continued support.

\end{document}